\documentclass{lmcs}
\pdfoutput=1

\usepackage{lastpage}
\lmcsdoi{17}{3}{1}
\lmcsheading{}{\pageref{LastPage}}{}{}%
{Sep.~30,~2019}{Jul.~20,~2021}{}

\usepackage{microtype}
\usepackage{amssymb}

\usepackage[utf8]{inputenc}
\usepackage{caption}
\usepackage{subcaption}
\usepackage{bm}
\usepackage{macros}
\usepackage{tikz}
\usetikzlibrary{automata,calc,patterns,positioning}
\theoremstyle{plain}

\begin{document}
\title[Affine Extensions of Integer Vector Addition Systems with States]{Affine Extensions of \\
  Integer Vector Addition Systems with States}
\titlecomment{{\lsuper*}A preliminary version of this paper appeared in the proceedings of the $29^\text{th}$ International Conference on Concurrency Theory (CONCUR), 2018~\cite{BlondinHM18}.}

\author[M.~Blondin]{Michael Blondin}
\address{Universit\'{e} de Sherbrooke, Canada}
\email{michael.blondin@usherbrooke.ca}
\thanks{M.\ Blondin was supported by a Discovery Grant from the
  Natural Sciences and Engineering Research Council of Canada (NSERC),
  and by a Quebec--Bavaria project and a start-up grant funded by the
  Fonds de recherche du Québec -- Nature et technologies (FRQNT)}

\author[C.~Haase]{Christoph Haase}
\address{University of Oxford, United Kingdom}
\email{christoph.haase@cs.ox.ac.uk}

\author[F.~Mazowiecki]{Filip Mazowiecki}
\address{Max Planck Institute for Software Systems, Germany}
\email{filipm@mpi-sws.org}
\thanks{F.~Mazowiecki's research has been carried out with financial
  support from
  the French State, managed by the French National Research Agency
  (ANR) in the frame of the ``Investments for the future'' Programme
  IdEx Bordeaux (ANR-10-IDEX-03-02).}

\author[M.~Raskin]{Mikhail Raskin}
\address{Technische Universit\"{a}t M\"{u}nchen, Germany}
\email{raskin@in.tum.de}
\thanks{M.\ Raskin was supported by funding from the European Research Council (ERC) under the European Union’s Horizon 2020 research and innovation programme under grant agreement No 787367 (PaVeS)}

\begin{abstract}
  We study the reachability problem for affine $\Z$-VASS, which are
  integer vector addition systems with states in which transitions
  perform affine transformations on the counters. This problem is
  easily seen to be undecidable in general, and we therefore restrict
  ourselves to affine $\Z$-VASS with the finite-monoid property
  (afmp-$\Z$-VASS). The latter have the property that the monoid
  generated by the matrices appearing in their affine transformations
  is finite. The class of afmp-$\Z$-VASS encompasses classical
  operations of counter machines such as resets, permutations,
  transfers and copies. We show that reachability in an afmp-$\Z$-VASS
  reduces to reachability in a $\Z$-VASS whose control-states grow
  linearly in the size of the matrix monoid. Our construction shows
  that reachability relations of afmp-$\Z$-VASS are semilinear, and in
  particular enables us to show that reachability in $\Z$-VASS with
  transfers and $\Z$-VASS with copies is PSPACE-complete. We then
  focus on the reachability problem for affine $\Z$-VASS with
  monogenic monoids: (possibly infinite) matrix monoids generated by a
  single matrix. We show that, in a particular case, the reachability
  problem is decidable for this class, disproving a conjecture about
  affine $\Z$-VASS with infinite matrix monoids we raised in a
  preliminary version of this paper. We complement this result by
  presenting an affine $\Z$-VASS with monogenic matrix monoid
  and undecidable reachability relation.
\end{abstract}

\maketitle

\section{Introduction}

\emph{Vector addition systems with states (VASS)} are a fundamental
model of computation comprising a finite-state controller with a
finite number of counters ranging over the natural numbers. When a
transition is taken, a counter can be incremented or decremented
provided that the resulting counter value is greater than or equal to
zero. Since the counters of a VASS are unbounded, a VASS gives rise to
an infinite transition system. One of the biggest advantages of VASS
is that most of the standard decision problems such as configuration
reachability and coverability are
decidable~\cite{KM69,May84,Kos82,Ler12}. Those properties make VASS
and their extensions a prime choice for reasoning about and modelling
concurrent, distributed and parametrised systems, see \eg\ the recent
surveys by Abdulla and Delzanno~\cite{AD16,Del16}.

In order to increase their modelling power, numerous extensions of
plain VASS have been proposed and studied in the literature over the
last 25 years. Due to the infinite-state nature of VASS, even minor
extensions often cross the undecidability frontier. For example, while
in the extension of VASS with hierarchical zero-tests on counters both
reachability and coverability remain decidable~\cite{Rei08,Bon13}, all
important decision problems for VASS with two counters which can
arbitrarily be tested for zero are undecidable~\cite{M67}. Another
example is the extension of VASS with reset and transfer operations. In a
\emph{reset VASS}, transitions may set a counter to zero, whereas
\emph{transfer VASS} generalize reset VASS and allow transitions to
move the contents of a counter onto another. While it was initially
widely believed that any extension of VASS either renders both
reachability and coverability undecidable, reset and transfer VASS
have provided an example of an extension which leads to an undecidable
reachability~\cite{AK76} yet decidable coverability
problem~\cite{DFS98}. Nevertheless, the computational costs for those
extensions are high: while coverability is EXPSPACE-complete for
VASS~\cite{Lip76,Rac78}, it becomes Ackermann-complete in the presence of
resets and transfers~\cite{Sch10,FFSS11}. For practical purposes, the
extension of VASS with transfers is particularly useful since transfer
VASS allow for reasoning about broadcast protocols and multithreaded
non-recursive C programs~\cite{EN98,KKW14}. It was already observed
in~\cite{EN98} that transfer VASS can be viewed as an instance of
so-called \emph{affine VASS}. An affine VASS is a generalization of VASS with
transitions labelled by pairs $(\mat A, \vec b)$, where $\mat A$ is a
$d\times d$ matrix over the integers and $\vec b \in \Z^d$ is an
integer vector. A transition switches the control-state while updating
the configuration of the counters $\vec v\in \N^d$ to $\mat A\cdot
\vec v + \vec b$, provided that $\mat A\cdot \vec v + \vec b\ge \vec
0$; otherwise, the transition is blocked. Transfer VASS can be viewed
as affine VASS in which the columns of all matrices are
$d$-dimensional unit vectors~\cite{EN98}.

Due to the symbolic state-explosion problem and Ackermann-hardness of
coverability, standard decision procedures for transfer VASS such as
the backward algorithm~\cite{ACJT96} do not \emph{per se} scale to
real-world instances. In recent years, numerous authors have proposed
the use of over-approximations in order to attenuate the symbolic
state-explosion problem for VASS and some of their extensions (see,
\eg,~\cite{ELMMN14,ALW16,BH17}). Most commonly, the basic idea is to
relax the integrality or non-negativity condition on the counters and
to allow them to take values from the non-negative rational numbers or
the integers. The latter class is usually referred to as $\Z$-VASS,
see \eg~\cite{HH14}. It is easily seen that if a configuration
is not reachable under the relaxed semantics, then the configuration
is also not reachable under the standard semantics. Hence, those
state-space over-approximations can, for instance, be used to prune
search spaces and empirically drastically speedup classical
algorithms for VASS such as the backward-algorithm. In this paper, we
investigate reachability in \emph{integer over-approximations} of
affine VASS, \ie, affine VASS in which a configuration of the counters
is a point in $\Z^d$, and in which all transitions are
non-blocking. Subsequently, we refer to such VASS as \emph{affine
  $\Z$-VASS}.

\parag{Main contributions.} We focus on affine $\Z$-VASS with the
\emph{finite-monoid property} (afmp-$\Z$-VASS), \ie\ where the matrix
monoid generated by all matrices occurring along transitions in the
affine $\Z$-VASS is finite. By a reduction to reachability in
$\Z$-VASS, we obtain decidability of reachability for the whole class
of afmp-$\Z$-VASS and semilinearity of their reachability relations.

In more detail, we show that reachability in an afmp-$\Z$-VASS can be
reduced to reachability in a $\Z$-VASS whose size is polynomial in the
size of the original afmp-$\Z$-VASS and in the norm of the finite
monoid $\monoid$ generated by the matrices occurring along
transitions, denoted by $\norm{\monoid}$. For a vast number of classes
of affine transformations considered in the literature,
$\norm{\monoid}$ is bounded exponentially in the dimension of the
matrices. This enables us to deduce a general PSPACE upper bound for
extensions of $\Z$-VASS such as transfer $\Z$-VASS and copy
$\Z$-VASS. By a slightly more elaborated analysis of this
construction, we are also able to provide a short proof of the already
known NP upper bound for reset $\Z$-VASS~\cite{HH14}. We also show
that a PSPACE lower bound of the reachability problem already holds
for the extension of $\Z$-VASS that only use permutation matrices in
their transition updates. This in turn gives PSPACE-completeness of
interesting classes such as transfer $\Z$-VASS and copy $\Z$-VASS.

Finally, we show that an affine $\Z$-VASS that has both transfers and
copies may not have the finite-monoid property, and that the
reachability problem for this class becomes undecidable. We complement
this result by investigating the case of monogenic classes,
\ie\ classes of monoids with a single generator. We show that although
reachability can still be undecidable for an affine $\Z$-VASS
with a monogenic matrix monoid, there exists a monogenic class
without the finite-monoid property for which reachability is
decidable.

All complexity results obtained in this paper are summarized in
Figure~\ref{fig:classification}, except for the undecidability of
general monogenic classes as it is a family of classes rather than one class.

\parag{Related work.} Our work is primarily related to the work of
Finkel and Leroux~\cite{FL02}, Iosif and Sangnier~\cite{IS16}, Haase
and Halfon~\cite{HH14}, and Cadilhac, Finkel and
McKenzie~\cite{CadilhacFM12,CadilhacFM13}. In~\cite{FL02}, Finkel and
Leroux consider a model more general than affine $\Z$-VASS in which
transitions are additionally equipped with guards which are Presburger
formulas defining admissible sets of vectors in which a transition
does not block. Given a sequence of transitions $\sigma$, Finkel and
Leroux show that the reachability set obtained from repeatedly
iterating $\sigma$, \ie, the \emph{acceleration} of $\sigma$, is
definable in Presburger arithmetic. Note that the model of Finkel and
Leroux does not allow for control-states and the usual tricks of
encoding each control-state by a counter or all control-states into
three counters~\cite{HP79} do not work over $\Z$ since transitions are
non-blocking. Iosif and Sangnier~\cite{IS16} investigated the
complexity of model checking problems for a variant of the model of
Finkel and Leroux with guards defined by convex polyhedra and with
control-states over a flat structure. Haase and Halfon~\cite{HH14}
studied the complexity of the reachability, coverability and inclusion
problems for $\Z$-VASS and reset $\Z$-VASS, two submodels of the
affine $\Z$-VASS that we study in this paper.  In
~\cite{CadilhacFM12,CadilhacFM13}, Cadilhac, Finkel and McKenzie
consider an extension of Parikh automata to affine Parikh automata
with the finite-monoid restriction like in our paper. These are
automata recognizing boolean languages, but the finite-monoid
restriction was exploited in a similar way to obtain some decidability
results in that context.  We finally remark that our models capture
variants of cost register automata that have only one $+$
operation~\cite{AR13,AFR14}.
  
\parag{Structure of the paper.} We introduce general notations and
affine $\Z$-VASS in Section~\ref{sec:preliminaries}. In
Section~\ref{sec:finite:monoid}, we give the reduction from
afmp-$\Z$-VASS to $\Z$-VASS. Subsequently, in
Section~\ref{sec:semilinear} we show that afmp-$\Z$-VASS have
semilinear reachability relations and discuss semilinearity of affine
$\Z$-VASS in general. In Section~\ref{sec:upper:bounds}, we show
PSPACE and NP upper bounds of the reachability problem for some
classes of afmp-$\Z$-VASS; and in Section~\ref{sec:lower:bounds} we
show PSPACE-hardness and undecidability results for some classes of
affine $\Z$-VASS. In Section~\ref{sec:inf:monoid}, we show that
reachability is undecidable for monogenic affine $\Z$-VASS and remains
decidable for a specific class of infinite monoids. Some concluding
remarks will be made in Section~\ref{sec:conclusion}.

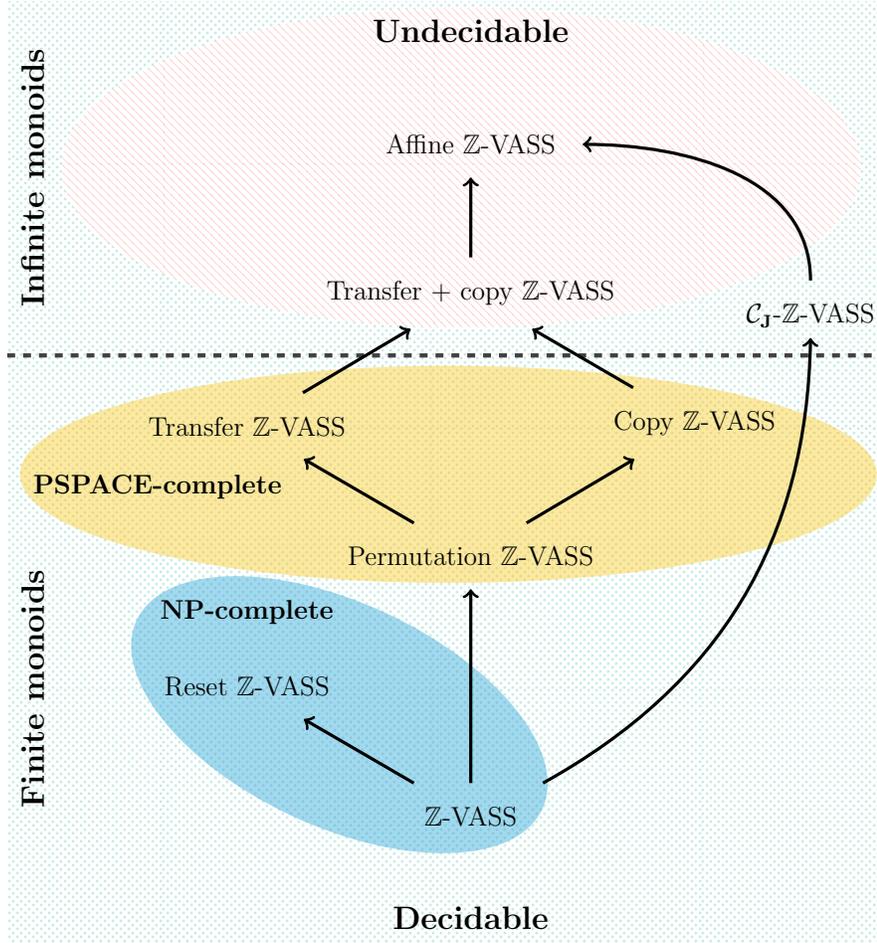
\begin{figure}[!h]
  \definecolor{colundecidable}{RGB}{255, 0, 0}
  \definecolor{colpspace}{RGB}{253, 199, 17}
  \definecolor{coldecidable}{RGB}{30, 176, 149}
  \definecolor{colnp}{RGB}{25, 158, 213}
  \centering
  \begin{tikzpicture}[auto, very thick, scale=0.85, transform shape]
    \tikzset{every node/.style={inner sep=10pt, align=center, minimum
        width=150pt, font=\large}}
    
    \fill[pattern=crosshatch dots, pattern color=coldecidable,
      opacity=0.5] (-6,-2) -- (-6,12.75) -- (7.75,12.75) -- (7.75,-2) --
    cycle;

    \fill[white, rotate around={0:(0,0)}] (1.1,10.15) ellipse
    (6.25cm and 2.5cm);
    \fill[pattern=north west lines, pattern color=colundecidable,
      opacity=0.4, rotate around={0:(0,0)}] (1.1,10.15) ellipse
    (6.25cm and 2.5cm);

    \fill[colnp, opacity=0.4, rotate around={-25:(0,0)}] (-1.4,1.1) ellipse
    (3.5cm and 1.75cm);

    \fill[colpspace, opacity=0.4, rotate around={0:(0,0)}] (0.9,5.35)
    ellipse (6.7cm and 1.7cm);

    \node[] (standard) [xshift=1.25cm] {$\Z$-VASS};
    
    \node[] (reset) [above=1cm of standard, xshift=-3.5cm] {Reset
      $\Z$-VASS};

    \node[] (resetpermutation) [above=1cm of reset, xshift=3.5cm]
         {Permutation $\Z$-VASS};

    \node[] (transfer) [above=1cm of resetpermutation, xshift=-3.5cm]
         {Transfer $\Z$-VASS};
    
    \node[] (copy) [above=1cm of resetpermutation, xshift=3.5cm] {Copy
      $\Z$-VASS};

    \node[] (transfercopy) [above=1cm of transfer, xshift=3.5cm]
         {Transfer + copy $\Z$-VASS};

    \node[] (affine) [above=1.25cm of transfercopy] {Affine $\Z$-VASS};

    \node[] (classone) [right=0cm of transfercopy, yshift=-10pt] {$\classone$-$\Z$-VASS};
    
    \node[rotate around={0:(0,0)}] at (-2.25,3.2) {\large
      \textbf{NP-complete}};

    \node[rotate around={0:(0,0)}] at (-3.65,5.15) {\large
      \textbf{PSPACE-complete}};

    \draw[dashed, black!50!gray, ultra thick] ($(transfercopy) +
    (-7.25,-0.975)$) -- ($(transfercopy) + (6.5,-0.975)$);

    \node[rotate around={90:(0,0)}] at (-5.6,2)
         {\Large \textbf{Finite monoids}};

    \node[rotate around={90:(0,0)}] at (-5.6,10)
         {\Large \textbf{Infinite monoids}};

    \node[] [below=15pt of standard]  {\Large \textbf{Decidable}};

    \node[] [above=20pt of affine]  {\Large \textbf{Undecidable}};

    \path[->]
    (standard) edge [] node {} (reset)
    (standard) edge (resetpermutation)
    (resetpermutation) edge [] node {} (transfer)
    (resetpermutation) edge [] node {} (copy)
    (transfer) edge [] node {} (transfercopy)
    (copy) edge [] node {} (transfercopy)
    (transfercopy) edge [] node {} (affine)
    ;

    \path[->]
    (standard) edge [bend right] node {} ($(classone)+(0,-0.35)$)
    (classone) edge [out=90, in=0] node {} ($(affine)+(1.75,0)$)
    ;
  \end{tikzpicture}
  \caption{Classification of the complexity of reachability in affine
    $\Z$-VASS in terms of classes of matrices. The rectangular regions
    below and above the horizontal dashed line correspond to classes
    of matrices with finite and infinite monoids respectively. The
    rectangular green dotted region and the elliptical red striped
    region correspond to the classes where reachability is decidable
    and undecidable, respectively. The elliptical blue region and the
    orange elliptical region correspond to the classes where
    reachability is NP-complete and PSPACE-complete respectively. The
    term ``$\classone$-$\Z$-VASS'' refers to the specific monogenic
    class of infinite monoids that will be defined in
    Section~\ref{sec:inf:monoids}.}\label{fig:classification}
\end{figure}

\section{Preliminaries}
\label{sec:preliminaries}

\parag{General notation.} For $n \in \N$, we write $[n]$ for the set
$\{1, 2, \ldots, n\}$. For every $\vec{x} = (x_1, x_2, \ldots, x_d)
\in \Z^d$ and every $i \in [d]$, we define $\vec{x}(i) \defeq x_i$. We
denote the identity matrix and the zero-vector by $\mat{I}$ and
$\vec{0}$ in every dimension, as there will be no ambiguity. For 
$\vec{x} \in \Z^d$ and $\mat{A} \in \Z^{d \times d}$, we define the
\emph{max-norm} of $\vec{x}$ and $\mat{A}$ as $\norm{\vec{x}} \defeq
\max\{|\vec{x}(i)| : i \in [d]\}$ and $\norm{\mat{A}} \defeq
\max\{\norm{\mat{A}_i} : i \in [d]\}$ where $\mat{A}_i$ denotes the
$i^\text{th}$ column of $\mat{A}$. We naturally extend this notation
to finite sets, \ie\ $\norm{G} \defeq \max\{\norm{\mat{A}} : \mat{A}
\in G\}$ for every $G \subseteq_\text{fin} \Z^{d \times d}$. We assume
that numbers are represented in binary, hence the entries of vectors
and matrices can be exponential in the size of their encodings.

\parag{Affine Integer VASS.} An \emph{affine integer vector addition
  system with states (affine $\Z$-VASS)} is a tuple $\V = (d, Q, T)$
where $d \in \N$, $Q$ is a finite set and $T \subseteq Q \times \Z^{d
  \times d} \times \Z^d \times Q$ is finite. Let us fix such a
$\V$. We call $d$ the \emph{dimension} of $\V$ and the elements of $Q$
and $T$ respectively \emph{control-states} and \emph{transitions}. For
every transition $t = (p, \mat{A}, \vec{b}, q)$, we define $\src{t} \defeq
p$, $\tgt{t} \defeq q$, $\tmat[t] \defeq \mat{A}$ and $\tvec[t] \defeq
\vec{b}$, and let $f_t \colon \Z^d \to \Z^d$ be the affine
transformation defined by $f_t(\vec{x}) = \mat{A} \cdot \vec{x} +
\vec{b}$. The \emph{size} of $\V$, denoted $|\V|$, is the number of
bits used to represent $d$, $Q$ and $T$ with coefficients written in
binary. For our purposes, we formally define it in a crude way as
$|\V| \defeq d + |Q| + (d^2 + d) \cdot |T| \cdot \max(1, \lceil
\log(\norm{T} + 1) \rceil)$ where $$\norm{T} \defeq
\max(\max\{\norm{\tvec[t]} : t \in T \}, \max\{\norm{\tmat[t]} : t \in
T \}).$$

A \emph{configuration} of $\V$ is a pair $(q, \vec{v}) \in Q \times
\Z^d$ which we write as $q(\vec{v})$. For every $t \in T$ and
$p(\vec{u}), q(\vec{v}) \in Q \times \Z^d$, we write $p(\vec{u})
\step{t} q(\vec{v})$ whenever $p = \src{t}$, $q = \tgt{t}$ and $\vec{v} =
f_t(\vec{u})$. We naturally extend $\step{}$ to sequences of
transitions as follows. For every $w = w_1 \cdots w_k \in T^k$ and $p(\vec{u}),
q(\vec{v}) \in Q \times \Z^d$, we write $p(\vec{u}) \step{w}
q(\vec{v})$ if either $k = 0$ (denoted $w = \varepsilon$) and $p(\vec{u}) = q(\vec{v})$, or
$k > 0$ and there exist $p_0(\vec{u}_0), p_1(\vec{u}_1), \ldots,
p_k(\vec{u}_k) \in Q \times \Z^d$ such that
$$
p(\vec{u}) = p_0(\vec{u}_0) \step{w_1} p_1(\vec{u}_1) \step{w_2} {}
\cdots {} \step{w_k} p_k(\vec{u}_k) = q(\vec{v}).
$$ We write $p(\vec{u}) \step{*} q(\vec{v})$ if there exists some $w \in
T^*$ such that $p(\vec{u}) \step{w} q(\vec{v})$. The relation
$\step{*}$ is called the \emph{reachability relation} of $\V$. If
$p(\vec{u}) \step{w} q(\vec{v})$, then we say that $w$ is a \emph{run
  from $p(\vec{u})$ to $q(\vec{v})$}, or simply a \emph{run} if the
source and target configurations are irrelevant. We also say that $w$
is a \emph{path} from $p$ to $q$, and if $p = q$ then we say that $w$
is a \emph{cycle}.

Let $\tmat[\V] \defeq \{\tmat[t] : t \in T\}$ and $\tvec[\V] \defeq
\{\tvec[t] : t \in T\}$. If $\V$ is clear from the context, we
sometimes simply write $\tmat$ and $\tvec$. The \emph{monoid of $\V$},
denoted $\monoid[\V]$ or sometimes simply $\monoid$, is the monoid
generated by $\tmat[\V]$, \ie\ it is the smallest set that contains
$\tmat[\V]$, is closed under matrix multiplication, and contains the
identity matrix. We say that a matrix $\mat{A} \in \N^{d \times d}$ is
respectively a (i)~\emph{reset}, (ii)~\emph{permutation},
(iii)~\emph{transfer}, (iv)~\emph{copyless}, or (v)~\emph{copy} matrix
if $\mat{A} \in \{0, 1\}^{d \times d}$ and
\begin{enumerate}[(i)]
\item $\mat{A}$ does not contain any $1$ outside of its diagonal;
  
\item $\mat{A}$ has exactly one $1$ in each row and each column;

\item $\mat{A}$ has exactly one $1$ in each column;

\item $\mat{A}$ has at most one $1$ in each column;

\item $\mat{A}$ has exactly one $1$ in each row.
\end{enumerate}

Analogously, we say that $\V$ is respectively a \emph{reset},
\emph{permutation}, \emph{transfer}, \emph{copyless}, or \emph{copy}
$\Z$-VASS if all matrices of $\tmat[\V]$ are reset, permutation,
transfer, copyless, or copy matrices. The monoids of such affine
$\Z$-VASS are finite and respectively of size at most $2^d$, $d!$,
$d^d$, $(d+1)^d$ and $d^d$. Copyless $\Z$-VASS correspond to a model
of copyless cost-register automata studied in~\cite{AFR14} (see the
remark below). If $\tmat[\V]$ only contains the identity matrix, then
$\V$ is simply called a $\Z$-VASS.

A \emph{class of matrices} $\class$ is a union $\bigcup_{d \geq 1}
\class_d$ where $\class_d$ is a finitely generated, but possibly
infinite, submonoid of $\N^{d \times d}$ for every $d \geq 1$. We say
that $\V$ belongs to a class $\class$ of $\Z$-VASS if $\monoid_\V
\subseteq \class$. If each $\class_d$ is finite, then we say that this
class of affine $\Z$-VASS has the \emph{finite-monoid property}
(afmp-$\Z$-VASS). For two classes $\class$ and $\class'$ we write
$\class + \class'$ to denote the smallest set $\mathcal{D} =
\bigcup_{d \geq 1} \mathcal{D}_d$ such that $\mathcal{D}_d$ is a
monoid that contains both $\class_d$ and $\class_d'$ for every $d \geq
1$. Note that this operation does not preserve finiteness. For
example if $\class$ and $\class'$ are the classes of transfer and copy matrices, respectively,
then $\class + \class'$ is infinite (see
Figure~\ref{fig:example} and Section~\ref{sec:lower:bounds}). We say
that a class $\class = \bigcup_{d \geq 1} \class_d$ is
\emph{nonnegative} if $\class_d \subseteq \N^{d \times d}$ for every
$d \geq 1$. We say that an affine $\Z$-VASS $\V$ is \emph{nonnegative} if
$\monoid_\V$ belongs to some nonnegative class of matrices. Note that the
classes of reset, permutation, transfer, copyless and copy matrices
are all nonnegative, respectively.

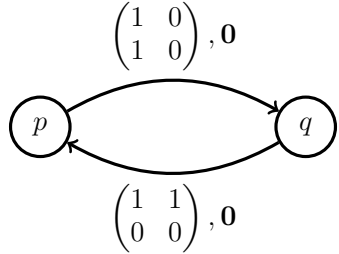
\begin{figure}[h]
  \centering
  \begin{tikzpicture}[->, node distance=3cm, auto, very thick, scale=0.9, transform shape, font=\large]
    \tikzset{every state/.style={inner sep=1pt, minimum size=25pt}}
 
    \node[state] (p) {$p$};
    \node[state] (q) [right= of p] {$q$};
    
    \path[->]
    (p) edge[bend left] node {
      $\begin{pmatrix}1 & 0 \\ 1 & 0\end{pmatrix}, \vec{0}$
    } (q)

    (q) edge[bend left] node {
      $\begin{pmatrix}1 & 1 \\ 0 & 0\end{pmatrix}, \vec{0}$
    } (p)
    ;
  \end{tikzpicture}
  \caption{Example of a transfer + copy $\Z$-VASS $\V$ which does not
    have the finite-monoid property.}\label{fig:example}
\end{figure}

We discuss the $\Z$-VASS $\V$ in Figure~\ref{fig:example} to give some
intuition behind the names transfer and copy $\Z$-VASS. The transition
from $p$ to $q$ is a copy transition and the transition from $q$ to
$p$ is a transfer transition. Notice that for every vector $(x,y) \in
\Z^2$, we have $p(x, y) \step{} q(x, x)$, \ie\ the value of the first
counter is copied to the second counter. Similarly, for the other
transition we have $q(x, y) \step{} p(x + y,0)$, that is the value of
the second counter is transferred to the first counter (resetting its
own value to $0$). Let $\mat{A}$ and $\mat{B}$ be the two matrices
used in $\V$. Note that $(\mat{A} \cdot \mat{B})^n$ is the matrix with
all entries equal to $2^{n-1}$, and hence $\monoid_\V$ is infinite.

\begin{rem}\label{remark:CRA}
  The variants of affine $\Z$-VASS that we consider are related to
  cost register automata (CRA) with only the $+$
  operation~\cite{AR13,AFR14} and without an output function. These
  are deterministic models with states and registers that upon reading
  an input, update their registers in the form $x \leftarrow y + c$,
  where $x, y$ are registers and $c$ is an integer. An affine
  $\Z$-VASS does not read any input, but is nondeterministic. Thus,
  one can identify an affine $\Z$-VASS with a CRA that reads sequences
  of transitions as words. In particular, the restrictions imposed on
  the studied CRAs correspond to copy $\Z$-VASS~\cite{AR13} and
  copyless $\Z$-VASS~\cite{AFR14}.
\end{rem}

\parag{Decision problems.} We consider the \emph{reachability} and the
\emph{coverability} problems parameterized by classes of matrices
$\class$:\bigskip

\noindent\;\;\underline{$\reach[\class]$ (reachability problem)} \\[2pt]
\begin{tabular}{ll}
  \textsc{Given}: & an affine $\Z$-VASS $\V = (d, Q, T)$ and
  configurations $p(\vec{u}), q(\vec{v})$ s.t.\ $\monoid[\V]
  \subseteq \class$. \\[2pt]
  
  \textsc{Decide}: & whether $p(\vec{u}) \step{*} q(\vec{v})$.
\end{tabular} \\[5pt]

\noindent\;\;\underline{$\cover[\class]$ (coverability problem)} \\[2pt]
\begin{tabular}{ll}
  \textsc{Given}: & an affine $\Z$-VASS $\V = (d, Q, T)$ and
  configurations $p(\vec{u}), q(\vec{v})$ s.t.\ $\monoid[\V]
  \subseteq \class$. \\[2pt]
  
  \textsc{Decide}: & whether there exists $\vec{v}' \in \Z^d$ such
  that $p(\vec{u}) \step{*} q(\vec{v}')$ and $\vec{v}' \geq \vec{v}$.
\end{tabular}\bigskip

For standard VASS (where configurations cannot hold negative values),
the coverability problem is much simpler than the reachability
problem. However, for affine $\Z$-VASS, these two problems coincide as
observed in~\cite[Lemma~2]{HH14}: the two problems are inter-reducible
in logarithmic space at the cost of doubling the number of
counters. Therefore we will only study the reachability problem in
this paper.

\section{From affine $\Z$-VASS with the finite-monoid property to $\Z$-VASS} \label{sec:finite:monoid}

The main result of this section is that every affine $\Z$-VASS $\V$
with the finite monoid property can be simulated by a $\Z$-VASS with twice the number of 
counters whose size is polynomial in $\norm{\monoid[\V]}$ and
$|\V|$. More formally, we show the following:

\begin{thm}\label{thm:monoid:to:zvass}
  For every afmp-$\Z$-VASS $\V = (d, Q, T)$ and $p,q \in Q$ there exist a $\Z$-VASS
  $\V' = (d', Q', T')$ and $p', q' \in Q'$ such that
  \begin{itemize}
  \item $d' = 2 \cdot d$,

  \item $|Q'| \leq 3 \cdot |\monoid[\V]| \cdot |Q|$,

  \item $|T'| \leq 4d \cdot |\monoid[\V]| \cdot (|Q| + |T|)$,

  \item $\norm{T'} \leq \norm{\monoid[\V]} \cdot \norm{T}$,

  \item $p(\vec{u}) \step{*} q(\vec{v})$ in $\V$ if and only if
    $p'(\vec{u}, \vec{0}) \step{*} q'(\vec{0}, \vec{v})$ in $\V'$.
  \end{itemize}
  Moreover, $\V'$, $p'$ and $q'$ are effectively computable from $\V$.
\end{thm}

\begin{cor}
  The reachability problem for afmp-$\Z$-VASS is decidable.
\end{cor}

\begin{proof}
  By Theorem~\ref{thm:monoid:to:zvass}, it suffices to construct, for
  a given afmp-$\Z$-VASS $\V$, the $\Z$-VASS $\V'$ and to test for
  reachability in $\V'$. It is known that reachability for $\Z$-VASS
  is in NP~\cite{HH14}. To effectively compute $\V'$ it suffices to
  provide a bound for $\norm{\monoid[\V]}$. It is known that if
  $|\monoid_\V|$ is finite then it is bounded by a computable function
  (see~\cite{MS77}), and hence $\norm{\monoid[\V]}$ is also
  computable.
\end{proof}

For the remainder of this section, let us fix some affine $\Z$-VASS
$\V$ such that $\monoid[\V]$ is finite. We proceed as follows to prove
Theorem~\ref{thm:monoid:to:zvass}. First, we introduce some notations
and intermediary lemmas characterizing reachability in affine
$\Z$-VASS. Next, we give a construction that essentially proves the
special case of Theorem~\ref{thm:monoid:to:zvass} where the initial
configuration is of the form $p(\vec{0})$. Finally, we prove
Theorem~\ref{thm:monoid:to:zvass} by extending this construction to
the general case.

It is worth noting that proving the general case is not necessary if
one is only interested in deciding reachability. Indeed, an initial
configuration $p(\vec{v})$ can be turned into one of the form
$p'(\vec{0})$ by adding a transition that adds $\vec{v}$. The reason
for proving the general case is that it establishes a stronger
relation that allows us to prove semilinearity of afmp-$\Z$-VASS
reachability relations in Section~\ref{sec:semilinear}.

\subsection{A characterization of reachability}
  
For every $w \in T^*$, $t \in T$ and $\vec{u} \in \Z^d$, let
\begin{align*}
  \tmat[\varepsilon] &\defeq \mat{I},
  & \varepsilon(\vec{u}) &\defeq \vec{u},
  \\ \tmat[w t] &\defeq \tmat[t] \cdot \tmat[w],
  & w t(\vec{u}) &\defeq \tmat[t] \cdot w(\vec{u}) + \tvec[t].
\end{align*}
Intuitively, for any sequence $w \in T^*$, $w(\vec{u})$ is the effect
of $w$ on $\vec{u}$, regardless of whether $w$ is an actual path of
the underlying graph. A simple induction yields the following
characterization:

\begin{lem} \label{lem:reach:charac}
  For every $w \in T^*$ and $p(\vec{u}), q(\vec{v}) \in Q \times
  \Z^d$, it is the case that $p(\vec{u}) \step{w} q(\vec{v})$ if and
  only if
  \begin{enumerate}[(a)]
  \item $w$ is a path from $p$ to $q$ in the underlying graph of
    $\V$, and \label{itm:path}
  \item $\vec{v} = w(\vec{u})$.
  \end{enumerate}
\end{lem}

Testing for reachability with Lemma~\ref{lem:reach:charac} requires
evaluating $w(\vec{u})$. This value can be evaluated conveniently as
follows:

\begin{lem} \label{lem:effect:decomp}
  For every $w=w_1w_2\cdots w_k \in T^k$ and $\vec{u} \in \Z^d$, the
  following holds:
  \begin{align}\label{lem:effect:equality}
    w(\vec{u}) = \tmat[w] \cdot \vec{u} + \sum_{i=1}^{k} \tmat(w_{i+1}
    w_{i+2} \cdots w_{k}) \cdot \tvec[w_i].
  \end{align}
  Moreover, $w(\vec{u}) = \tmat[w] \cdot \vec{u} + w(\vec{0})$.
\end{lem}

\begin{proof}[Proof of Lemma~\ref{lem:effect:decomp}]
  We prove~\eqref{lem:effect:equality} by induction on $k$. The base
  case follows from $\varepsilon(\vec{u}) = \vec{u} = \mat{I} \cdot
  \vec{u} + \vec{0} = \tmat[\varepsilon] \cdot \vec{u} +
  \vec{0}$. Assume that $k > 0$ and that the claim holds for sequences
  of length $k - 1$. For simplicity we denote $\sigma \defeq w_1
  \cdots w_{k-1}$. We have:
  \begin{align}
    w(\vec{u})
    &= \sigma w_k(\vec{u}) \nonumber \\
    &= \tmat[w_k] \cdot \sigma (\vec{u}) + \tvec[w_k] \label{eq:swu} \\
    &= \tmat[w_k] \cdot \left(\tmat[\sigma] \cdot \vec{u} +
    \sum_{i=1}^{k-1} \tmat(w_{i+1} w_{i+2} \cdots w_{k-1}) \cdot
    \tvec[w_i]\right) + \tvec[w_k] \label{eq:ind:hyp} \\
    &= \tmat[w_k] \cdot \tmat[\sigma] \cdot \vec{u} + \sum_{i=1}^{k-1}
    \tmat[w_k] \cdot \tmat(w_{i+1} w_{i+2} \cdots w_{k-1}) \cdot
    \tvec[w_i] + \tvec[w_k] \nonumber \\
    &= \tmat[\sigma w_k] \cdot \vec{u} + \sum_{i=1}^{k-1}
    \tmat(w_{i+1} w_{i+2} \cdots w_{k}) \cdot \tvec[w_i] + \tvec[w_k]
    \label{eq:mat:def} \\
    &= \tmat[w] \cdot \vec{u} + \sum_{i=1}^{k} \tmat(w_{i+1} w_{i+2}
    \cdots w_{k}) \cdot \tvec[w_i] \nonumber
  \end{align}
  where~\eqref{eq:swu}, \eqref{eq:ind:hyp} and~\eqref{eq:mat:def}
  follow respectively by definition of $\sigma w_k(\vec{u})$, by
  induction hypothesis and by definition of $M(\sigma w_k)$.
  
  The last part of the lemma follows from
  applying~\eqref{lem:effect:equality} to $w(\vec{0})$ and
  $w(\vec{u})$, and observing that subtracting them results in
  $w(\vec{u}) - w(\vec{0}) = \tmat[w] \cdot \vec{u}$.
\end{proof}

Observe that Lemma~\ref{lem:effect:decomp} is trivial for the
particular case of $\Z$-VASS. Indeed, we obtain $w(\vec{u}) = \vec{u}
+ \sum_{i=1}^{k} \tvec[w_i]$, which is the sum of transition vectors
as expected for a $\Z$-VASS.

\subsection{Reachability from the origin}

We make use of Lemmas~\ref{lem:reach:charac}
and~\ref{lem:effect:decomp} to construct a $\Z$-VASS $\V' = (d, Q',
T')$ for the special case of Theorem~\ref{thm:monoid:to:zvass} where
the initial configuration is of the form $p(\vec{0})$. The states and
transitions of $\V'$ are defined as:
\begin{align*}
  Q' &\defeq Q \times \monoid, \\
  T' &\defeq \{((\tgt{t}, \mat{A}), \mat{I}, \mat{A} \cdot \tvec[t],
  (\src{t}, \mat{A} \cdot \tmat[t])) : \mat{A} \in \monoid, t \in T\}.
\end{align*}

The idea behind $\V'$ is to simulate a path $w$ of $\V$ backwards and
to evaluate $w(\vec{0})$ as the sum identified in
Lemma~\ref{lem:effect:decomp}. More formally, $\V'$ and $\V$ are
related as follows:
  
\begin{prop} \label{prop:theorem_for0}
  \leavevmode
  \begin{enumerate}[(a)]
  \item For every $w \in T^*$, if $p(\vec{0}) \step{w} q(\vec{v})$ in
    $\V$, then $q'(\vec{0}) \step{*} p'(\vec{v})$ in $\V'$, where $q'
    \defeq (q, \mat{I})$ and $p' \defeq (p, \tmat[w])$.

  \item If $q'(\vec{0}) \step{*} p'(\vec{v})$ in $\V'$, where $q'
    \defeq (q, \mat{I})$ and $p' \defeq (p, \mat{A})$, then there
    exists $w \in T^*$ such that $\tmat[w] = \mat{A}$ and $p(\vec{0})
    \step{w} q(\vec{v})$ in $\V$.
  \end{enumerate}
\end{prop}

\begin{proof}
  (a)~By Lemma~\ref{lem:reach:charac}, $\V$ has a
  path $w \in T^*$ such that $w(\vec{0}) = \vec{v}$. Let $k \defeq
  |w|$. Let $\mat{A}_0 \defeq \mat{I}$, and for every $i \in [k]$ let
  \begin{align*}
    \mat{A}_i &\defeq \tmat[w_{k-i+1} \cdots w_{k-1} w_k], \\
    \vec{b}_i &\defeq \mat{A}_{i-1} \cdot \tvec[w_{k-i+1}], \\
    w_i' &\defeq ((\tgt{w_{k-i+1}}, \mat{A}_{i-1}), \mat{I}, \vec{b}_i,
    (\src{w_{k-i+1}}, \mat{A}_i)).
  \end{align*}
  We claim that $w' \defeq w_1' w_2' \cdots w_k'$ is such that $(q,
  \mat{A}_0) \step{w'} (p, \mat{A}_k)$ in $\V'$. Note that the
  validity of the claim completes the proof since $\mat{A}_0 =
  \mat{I}$ and $\mat{A}_k = \tmat[w]$.

  It follows immediately from the definition of $T'$ that $w_i' \in
  T'$ for every $i \in [k]$ and hence that $w'$ is a path from $(q,
  \mat{A}_0)$ to $(p, \mat{A}_k)$. By Lemma~\ref{lem:reach:charac}, it
  remains to show that $w'(\vec{0}) = \vec{v}$:
  \begingroup
  \allowdisplaybreaks
  \begin{align*}
    w'(\vec{0})
    &= \sum_{i=1}^{k} \tmat(w_{i+1}' w_{i+2}' \cdots w_{k}') \cdot
    \tvec[w_i']
    && \text{(by Lemma~\ref{lem:effect:decomp} applied to
      $w'(\vec{0})$)} \\
    &= \sum_{i=1}^{k} \tvec[w_i']
    && \text{(by $\tmat[w_i'] = \mat{I}$ for every $i \in [k]$)} \\
    &= \sum_{i=1}^{k} \mat{A}_{i-1} \cdot \tvec[w_{k-i+1}]
    && \text{(by definition of $\tvec[w_i']$)} \\
    &= \sum_{i=1}^{k} \tmat(w_{k-i+2} w_{k-i+1} \cdots w_k) \cdot
    \tvec[w_{k-i+1}]
    && \text{(by definition of $\mat{A}_{i-1}$)} \\
    &= \sum_{i=1}^{k} \tmat(w_{i+1} w_{i+2} \cdots w_k) \cdot
    \tvec[w_i]
    && \text{(by inspection of the sum)} \\
    &= w(\vec{0})
    && \text{(by Lemma~\ref{lem:effect:decomp} applied to $w(\vec{0})$).}
  \end{align*}
  \endgroup

  (b)~Similarly, by Lemma~\ref{lem:reach:charac}, there
  exists a path $w'$ of $\V'$ such that $w'(\vec{0}) = \vec{v}$, and
  it suffices to exhibit a path $w \in T^*$ from $p$ to $q$ in $\V$
  such that $w(\vec{0}) = \vec{v}$ and $\tmat[w] = \mat{A}$. Let $k
  \defeq |w'|$. For every $i \in [k]$, let $w_i' = ((p_i, \mat{A}_i),
  \mat{I}, \vec{b}_i, (q_i, \mat{B}_i))$. By definition of $T'$, for
  every $i \in [k]$, there exists a (possibly non unique) transition
  $t_i \in T$ such that $\tgt{t} = p_i$, $\src{t} = q_i$, $\vec{b}_i =
  \mat{A}_i \cdot \tvec[t]$ and $\mat{B}_i = \mat{A}_i \cdot
  \tmat[t]$. We set $w \defeq t_1 t_2 \cdots t_k$. It is readily seen
  that $w$ is a path from $p$ to $q$. To prove $w(\vec{0}) = \vec{v}$
  and $\tmat[w] = \mat{A}$, Lemma~\ref{lem:effect:decomp} can be
  applied as in the previous implication.
\end{proof}

\subsection{Reachability from an arbitrary configuration}

We now construct the $\Z$-VASS $\V''= (2d, Q'', T'')$ of
Theorem~\ref{thm:monoid:to:zvass} which is obtained mostly from
$\V'$. The states of $\V''$ are defined as:
\begin{align*}
  Q'' &\defeq Q \cup Q' \cup \overline{Q'}
  && \text{where}\quad \overline{Q'} \defeq \{(\overline{q}, \mat{A}) :
  (q, \mat{A}) \in Q'\}.
\end{align*}
To simplify notation, given two vectors $\vec{u}, \vec{v} \in
\Z^d$ we write $(\vec{u}, \vec{v})$ for the vector of $\Z^{2d}$ equal
to $\vec{u}$ on the first $d$ components and equal to $\vec{v}$ on the
last $d$ components. The set $T''$ consists of four disjoint subsets
of transitions $T_\text{simul} \cup T_\text{end} \cup T_\text{mult}
\cup T_\text{final}$ working in four sequential stages. Intuitively,
these transitions allow (1)~$\V''$ to simulate a path $w$ of $\V$
backwards in order to compute $w(\vec{0})$; (2)~guess the end of this
path; (3)~compute $\tmat[w] \cdot \vec{u}$ by using the fact that
$\tmat[w]$ is stored in its control-state; and (4) guess the end of
this matrix multiplication.

The first set of transitions is defined as:
\begin{align*}
  T_\text{simul} &\defeq \{(\src{t}, \mat{I}, (\vec{0}, \tvec[t]),
  \tgt{t}) : t \in T'\}.
\end{align*}
Its purpose is to simulate $T'$ on the last $d$ counters. The second
set is defined as:
\begin{align*}
  T_\text{end} &\defeq \{((q, \mat{A}), \mat{I}, (\vec{0}, \vec{0}),
  (\overline{q}, \mat{A})) : (q, \mat{A}) \in Q'\},
\end{align*}
and its purpose is to nondeterministically guess the end of a run in
$\V'$ by simply marking $q$. The third set is defined as:
\begin{alignat*}{2}
  T_\text{mult}
  &\defeq\ && \{((\overline{q}, \mat{A}), \mat{I}, (-\vec{e}_i,
  \mat{A} \cdot \vec{e}_i), (\overline{q}, \mat{A})) : (\overline{q},
  \mat{A}) \in \overline{Q'}, i \in [d]\} \cup {} \\
  &&& \{((\overline{q}, \mat{A}), \mat{I}, (\vec{e}_i, -\mat{A} \cdot
  \vec{e}_i), (\overline{q}, \mat{A})) : (\overline{q}, \mat{A}) \in
  \overline{Q'}, i \in [d]\},
\end{alignat*}
where $\vec{e}_i$ is the $i$-th unit vector such that $\vec{e}_i(i) = 1$
and $\vec e_i(j)=0$ for all $i\neq j$. The
purpose of $T_\text{mult}$ is to compute $\mat{A} \cdot \vec{u}$ from
the $d$ first counters onto the $d$ last counters. Finally,
$T_\text{final}$ is defined as:
\begin{align*}
 T_\text{final} &\defeq \{((\overline{q}, \mat{A}), \mat{I}, (\vec{0},
 \vec{0}), q) : (\overline{q}, \mat{A}) \in \overline{Q'}\},
\end{align*}
and its purpose is to guess the end of the matrix multiplication
performed with $T_\text{mult}$.

We may now prove Theorem~\ref{thm:monoid:to:zvass}:

\begin{proof}[Proof of Theorem~\ref{thm:monoid:to:zvass}]
  First, note that we obtain
  \begin{align*}
    |Q''|
    &= 2 \cdot |Q'| + |Q| \\
    &\leq 3 \cdot |Q| \cdot |\monoid|, \\[5pt]    
    |T''|
    &= |T'| + |Q'| + 2d \cdot |Q'| + |Q'| \\
    &= |T'| + 2(d+1) \cdot |Q'| \\
    &= |T| \cdot |\monoid| + 2(d+1) \cdot |Q| \cdot |\monoid| \\
    &\leq 4d \cdot |\monoid| \cdot (|T| + |Q|).
  \end{align*}
  Moreover, we have:
  \begin{align*}
    \norm{T''}
    &= \max(\norm{T'}, \norm{\monoid}) \\
    &\leq \max(\norm{\monoid} \cdot \norm{T}, \norm{\monoid}) \\
    &= \norm{\monoid} \cdot \norm{T}.
  \end{align*}

  We conclude by proving that $p(\vec{u}) \step{*} q(\vec{v})$ in $\V$
  if and only if $q'(\vec{u}, \vec{0}) \step{*} p(\vec{0}, \vec{v})$
  in $\V''$, where $q' \defeq (q, \mat{I})$.

  $\Rightarrow$) By Lemma~\ref{lem:reach:charac}, there exists a path
  $w$ of $\V$ such that $w(\vec{u}) = \vec{v}$. By definition of
  $T_\text{simul}$ and $T_\text{end}$, and by
  Proposition~\ref{prop:theorem_for0}, it is the case that
  $q'(\vec{u}, \vec{0}) \step{*} p'(\vec{u}, w(\vec{0}))$ where $p'
  \defeq (p, \tmat[w])$. The transitions of $T_\text{mult}$ allow to
  transform $(\vec{u}, w(\vec{0}))$ into $(\vec{0}, w(\vec{0}) +
  \tmat[w] \cdot \vec{u})$. Thus, using $T_\text{final}$, we can reach
  the configuration $p(w(\vec{0}) + \tmat[w] \cdot \vec{u})$. This
  concludes the proof since $w(\vec{u}) = w(\vec{0}) + \tmat[w] \cdot
  \vec{u}$ by Lemma~\ref{lem:effect:decomp}.

  $\Leftarrow$) The converse implication follows the same steps as
  the previous one. It suffices to observe that the first part of a
  run of $\V''$ defines the value $w(\vec{0})$, while the second part
  of the run defines $\tmat[w] \cdot \vec{u}$.
\end{proof}

\section{Semilinearity of affine $\Z$-VASS} \label{sec:semilinear}

A subset of $\Z^d$ is called \emph{semilinear} if it is definable by a
formula of Presburger arithmetic~\cite{Pr29}, \ie\ by a formula of
$\mathsf{FO}(\Z, +, <)$, the decidable first-order logic over $\Z$
with addition and order. Semilinear sets capture precisely finite
unions of sets of the form $\vec{b} + \N \cdot \vec{p}_1 + \N \cdot
\vec{p}_2 + \ldots + \N \cdot \vec{p}_k$ with each $\vec p_i\in \Z^d$,
and are effectively closed under basic operations such as finite sums,
intersection and complement. Those properties make semilinear sets an
important tool in many areas of computer science and find use whenever
infinite subsets of $\Z^d$ need to be manipulated.

The results of Section~\ref{sec:finite:monoid} enable us to show that
any affine $\Z$-VASS with the finite-monoid property has a semilinear
reachability relation:

\begin{thm} \label{thm:finite:semilin}
  Given an afmp-$\Z$-VASS $\V = (d, Q, T)$ and $p, q \in Q$, it is
  possible to compute an existential Presburger formula $\varphi_{\V,
    p, q}$ of size at most $\mathcal{O}(\mathrm{poly}(|\V|,
  |\monoid[\V]|, \log \norm{\monoid[\V]}))$ such that $\varphi_{\V, p,
    q}(\vec{u}, \vec{v})$ holds if and only if $p(\vec{u}) \step{*}
  q(\vec{v})$ in $\V$.
\end{thm}

\begin{proof}
  By Theorem~\ref{thm:monoid:to:zvass}, there exist an effectively
  computable $\Z$-VASS $\V' = (d', Q', T')$ and $p', q' \in Q'$ such
  that $d' = 2 \cdot d$, $|Q'| \leq 3 \cdot |\monoid| \cdot |Q|$,
  $|T'| \leq 4d \cdot |\monoid| \cdot (|Q| + |T|)$, $\norm{T'} \leq
  \norm{\monoid} \cdot \norm{T}$ and
  \begin{align}
  p(\vec{u}) \step{*} q(\vec{v}) \text{ in } \V \text{ if and only if
  } p'(\vec{u}, \vec{0}) \step{*} q'(\vec{0}, \vec{v}) \text{ in }
  \V'. \label{eq:equiv}
  \end{align}
  By~\cite[Sect.~3]{HH14}, we can compute an existential Presburger
  formula $\psi$ of linear size in $|\V'|$ such that $\psi(\vec{x},
  \vec{x}', \vec{y}, \vec{y}')$ holds if and only if $p'(\vec{x},
  \vec{x}') \step{*} q'(\vec{y}, \vec{y}')$ in
  $\V'$. By Equation~\eqref{eq:equiv}, the formula $\varphi_{\V, p, q}(\vec{x},
  \vec{y}) \defeq \psi(\vec{x}, \vec{0}, \vec{0}, \vec{y})$ satisfies
  the theorem.
\end{proof}

It was observed in~\cite{FL02,Boi98} that the reachability relation of
a $\Z$-VASS $\V = (d, Q, T)$, such that $|Q| = |\tmat[\V]| = 1$, is
semilinear if and only if $\monoid[\V]$ is
finite. Theorem~\ref{thm:finite:semilin} shows that if we do not bound
the number of states and matrices, \ie\ drop the assumption $|Q| =
|\tmat[\V]| = 1$, then ($\Leftarrow$) remains true. It is
natural to ask whether ($\Rightarrow$) also remains true.

\begin{figure}[h]
  \centering
  \begin{tikzpicture}[->, node distance=3cm, auto, very thick, scale=0.9, transform shape, font=\large]
    \tikzset{every state/.style={inner sep=1pt, minimum size=25pt}}
    
    \node[state] (p) {$p$};
    \node[state] (q) [right=5cm of p] {$p$};
    \node[state] (r) [right= of q] {$q$};
    
    \path[->]
    (p) edge[loop above] node {
      $\begin{pmatrix}1 & 1 \\ 1 & 1\end{pmatrix}, \vec{0}$
    } (p)
    
    (p) edge[in=10, out=40, loop] node[yshift=-10pt] {
      $\mat{I}, \begin{pmatrix}1 \\ 0\end{pmatrix}$
    } (p)      
    
    (p) edge[in=-40, out=-10, loop] node[yshift=10pt] {
      $\mat{I}, \begin{pmatrix}0 \\ 1\end{pmatrix}$
    } (p)      

    (p) edge[in=140, out=170, loop] node[yshift=-10pt] {
      $\mat{I}, \begin{pmatrix}-1 \\ 0\end{pmatrix}$
    } (p)      
    
    (p) edge[in=-170, out=-140, loop] node[yshift=10pt] {
      $\mat{I}, \begin{pmatrix}0 \\ -1\end{pmatrix}$
    } (p)      

    (q) edge[bend left] node {
      $\begin{pmatrix}1 & 1 \\ 1 & 1\end{pmatrix}, \vec{0}$
    } (r)

    (r) edge[bend left] node {
      $\begin{pmatrix}0 & 0 \\ 0 & 0\end{pmatrix}, \vec{0}$
    } (q)
    ;
  \end{tikzpicture}
  \caption{Examples of affine $\Z$-VASS with infinite monoids and
    semilinear reachability relations.}\label{fig:inf:semilinear}
\end{figure}
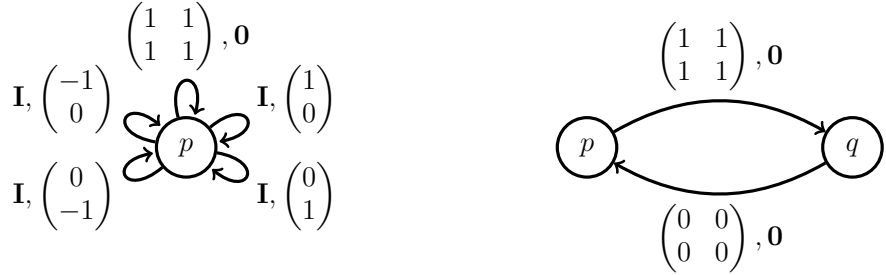

Let $\V_1$ and $\V_2$ be the affine $\Z$-VASS illustrated in
Figure~\ref{fig:inf:semilinear} from left to right respectively. Note
that $\monoid[\V_1]$ and $\monoid[\V_2]$ are both infinite due to the
matrix made only of $1$s. Moreover, the reachability relations of
$\V_1$ and $\V_2$ are semilinear since the former can reach any target
configuration from any initial configuration, and since the latter can
only generate finitely many vectors due to the zero matrix. Since
$\V_1$ has a single control-state, $|\tmat[\V_1]| = |\tmat[\V_2]| = 2$
and $\tvec[\V_2] = \{\vec{0}\}$, any simple natural extension of the
characterization of semilinearity in terms of the number of
control-states, matrices and vectors fails.

It is worth noting that an affine $\Z$-VASS with an infinite monoid
may have a non semilinear reachability relation. Indeed,
Figure~\ref{fig:example} depicts a transfer + copy $\Z$-VASS with an
infinite monoid and such that $\{\vec{v} : p(1, 1) \step{*}
q(\vec{v})\} = \{(2^n, 2^n) : n \in \N\}$, which is known to be non
semilinear. Moreover, this proves that even the reachability set from
$p(1, 1)$ is not semilinear.

\section{Complexity of reachability for afmp-$\Z$-VASS} \label{sec:upper:bounds}

In this section, we use the results of Section~\ref{sec:finite:monoid}
to show that reachability belongs to PSPACE for a large class of
afmp-$\Z$-VASS encompassing all variants discussed in
Section~\ref{sec:preliminaries}. Moreover, we give a novel proof to
the known NP membership of reachability for reset $\Z$-VASS.

For every finite set $G_d \subseteq \Z^{d \times d}$, let $\langle G_d
\rangle$ be the monoid generated by $G_d$. We have:

\begin{thm} \label{thm:pspace:upper}
  Let $\class = \bigcup_{d \geq 1} \class_d$ be a class of matrices
  such that $\class_d$ is finite for every $d \geq 1$. It is the case
  that $\reach[\class] \in \text{PSPACE}$ if there exists a polynomial
  $\mathrm{poly}$ such that $|\langle G_d \rangle| + \norm{\langle G_d
    \rangle} \leq 2^{\mathrm{poly}(d + \log\norm{G_d})}$ for every $d
  \geq 1$ and every finite set $G_d$ such that $\langle G_d \rangle
  \subseteq \class_d$.
\end{thm}

\begin{proof}
  Let $\V = (d, Q, T)$ be an affine $\Z$-VASS from class $\class$. Let
  $\V' = (d, Q', T')$ be the $\Z$-VASS obtained from $\V$ in
  Theorem~\ref{thm:monoid:to:zvass}. Recall that, by
  Theorem~\ref{thm:monoid:to:zvass}, $p(\vec{u}) \step{*} q(\vec{v})$
  in $\V$ if and only if $p'(\vec{u}, \vec{0}) \step{*} q'(\vec{0},
  \vec{v})$ in $\V'$. Therefore, it suffices to check the latter for
  determining reachability in $\V$.

  We invoke a result of~\cite{BFGHM15} on the flattability of
  $\Z$-VASS. By~\cite[Prop.~3]{BFGHM15}, $p'(\vec{x}) \step{*}
  q'(\vec{y})$ in $\V'$ if and only if there exist $k \leq |T'|$,
  $\alpha_0, \beta_1, \alpha_1, \ldots, \beta_k, \alpha_k \in (T')^*$
  and $\vec{e} \in \N^k$ such that
  \begin{enumerate}[(i)]
  \item $p'(\vec{x}) \step{\alpha_0 \beta_1^{\vec{e}(1)} \alpha_1 \cdots
    \beta_k^{\vec{e}(k)}} q'(\vec{y})$ in $\V'$, \label{itm:diophantine}

  \item $\beta_i$ is a cycle for every $i \in [k]$, and
    
  \item $\alpha_0 \beta_1 \alpha_1 \cdots \beta_k \alpha_k$ is a path
    from $p'$ to $q'$ of length at most $2 \cdot |Q'| \cdot
    |T'|$. \label{itm:flat:path}
  \end{enumerate}

  For every $w \in (T')^*$, let $\tvec[w] \defeq \sum_{i = 1}^{|w|}
  \tvec[w_i]$. By Lemma~\ref{lem:effect:decomp} (see the remark below
  the proof of Lemma~\ref{lem:effect:decomp}), we have $w(\vec{u}) =
  \vec{u} + \tvec[w]$ for every $\vec{u} \in \Z^d$. Thus, by
  Lemma~\ref{lem:reach:charac}, checking~\ref{itm:diophantine},
  assuming~\ref{itm:flat:path}, amounts to testing whether $\vec{e}$
  is a solution of the following system of linear Diophantine
  equations:
  \begin{align}
    \vec{x} +
    \sum_{i = 0}^k \tvec[\alpha_i] +
    \begin{pmatrix}
      \tvec[\beta_1] &
      \tvec[\beta_2] &
      \cdots &
      \tvec[\beta_k]
    \end{pmatrix} \cdot \vec{e}
    &= \vec{y}. \label{eq:diophantine}
  \end{align}
  Let $G_d \defeq \tmat[\V]$. Note that $\norm{G_d} \leq \norm{T}$ and
  that $\langle G_d \rangle = \monoid[\V]$. Let $m \defeq 2 \cdot |Q'|
  \cdot |T'|$. By Theorem~\ref{thm:monoid:to:zvass}, we have $m \leq
  48d \cdot |\monoid|^2 \cdot |Q|^2 \cdot |T|$. Thus, since $\langle
  G_d \rangle$ is a submonoid of $\class_d$, and by assumption on
  $\class_d$, we have $$m \leq 48d \cdot
  \left(2^{\mathrm{poly}(d + \log\norm{T})}\right)^2 \cdot |Q|^2 \cdot
  |T|.$$ Thus, $m$ is exponential in $|\V|$.
  
  We describe a polynomial-space non deterministic Turing machine
  $\mathcal{A}$ for testing whether $p'(\vec{x}) \step{*} q'(\vec{y})$
  in $\V'$. The proof follows from NPSPACE = PSPACE. Machine
  $\mathcal{A}$ guesses $k \leq |T'|$, a path $\pi = \alpha_0 \beta_1
  \alpha_1 \cdots \beta_k \alpha_k$ of length at most $m$ from $p'$ to
  $q'$, and $\vec{e} \in \N^k$, and tests
  whether~\eqref{eq:diophantine} holds for $\pi$. Note that we are not
  given $\V'$, but $\V$, so we must be careful for the machine to work
  in polynomial space.

  Instead of fully constructing $\V'$ and fully guessing $\pi$, we do
  both on the fly, and also construct $\tvec[\alpha_0],
  \tvec[\beta_1], \ldots, \tvec[\beta_k], \tvec[\alpha_k]$ on the fly
  as partial sums as we guess $\pi$. Note that to ensure that each
  $\beta_i$ is a cycle, we do not need to fully store $\beta_i$ but
  only its starting control-state. Moreover, note that
  $\norm{\tvec[\alpha_i]}, \norm{\tvec[\beta_i]} \leq m \cdot
  \norm{T'}$ for every $i$. By Theorem~\ref{thm:monoid:to:zvass} and
  by assumption on $\class_d$, we have
  \begin{align*}
    \norm{T'}
    &\leq \norm{\langle G_d \rangle} \cdot \norm{T} \\
    &\leq 2^{\mathrm{poly}(d + \log \norm{T})} \cdot \norm{T}.
  \end{align*}
  Hence, each $\alpha_i$ and $\beta_i$ has a binary representation of
  polynomial size in $|\V|$.

  By~\cite[Prop.~4]{CH16}, \eqref{eq:diophantine} has a solution if
  and only if it has a solution $\vec{e} \in \N^k$ such
  that $$\norm{\vec{e}} \leq \left((k + 1) \cdot
  \max\{\norm{\tvec[\beta_i]} : i \in [k]\} + \norm{\vec{x}} +
  \norm{\vec{y}} + \sum_{i=0}^k \norm{\tvec[\alpha_i]} +
  1\right)^{d'}.$$ Since $d' = 2 \cdot d$, this means that we can
  guess a vector $\vec{e} \in \N^k$ whose binary representation is of
  polynomial size, and that we can thus
  evaluate~\eqref{eq:diophantine} in polynomial time.
\end{proof}

\begin{cor}
  The reachability problem for nonnegative afmp-$\Z$-VASS is in
  PSPACE, and hence in particular for reset, permutation, transfer,
  copy and copyless $\Z$-VASS.
\end{cor}

\begin{proof}
  Let $\class = \bigcup_{d \geq 1} \class_d$ be a class of nonnegative
  matrices. Let $d \geq 1$ and let $G_d$ be a finite set of matrices
  such that $\langle G_d \rangle \subseteq \class_d$. By~\cite[Theorem
    A.2]{WS91}, whose proof appears in~\cite{We87} written by one of
  the same authors, we have:
  \begin{alignat*}{2}
    |\langle G_d \rangle| &\leq \norm{G_d}^{d^2 \cdot (d - 1)} \cdot
    5^{d^3 / 2} \cdot {d^{d^3}} \cdot d^2 &&= 2^{d^2 \cdot (d-1) \cdot
      \log \norm{G_d} + (d^3 / 2) \cdot \log 5 + d^3 \cdot \log d + 2
      \cdot \log d},  \\
    \norm{\langle G_d \rangle} &\leq \norm{G_d}^{d-1} \cdot 5^{d/2}
    \cdot d^d &&= 2^{(d-1) \cdot \log \norm{G_d} + (d / 2) \cdot \log 5
      + d \cdot \log d}.
  \end{alignat*}
  Thus, $\class$ satisfies the requirements of
  Theorem~\ref{thm:pspace:upper}. To complete the proof, observe that
  determining whether $\monoid[\V]$ is finite can be done in time
  $\mathcal{O}(d^6 \cdot |T|)$, again by~\cite[Theorem A.2]{WS91}
  and~\cite{We87}.

  Note that this proof applies to reset, permutation, transfer, copy
  and copyless classes, respectively, as they are all
  nonnegative. However, there is a much simpler argument for these
  specific classes. Indeed, their matrices all have a max-norm of a
  most $1$ and thus their monoids contain at most $2^{d^2}$ matrices.
\end{proof}

\begin{thmC}[\cite{HH14}]\label{thm:reset:NP}
  The reachability problem for reset $\Z$-VASS belongs to NP.
\end{thmC}

\begin{proof}
  Let $\V = (d, Q, T)$ be a reset $\Z$-VASS. The proof does not follow
  immediately from Theorem~\ref{thm:monoid:to:zvass} because
  $\monoid[\V]$ can be of size up to $2^d$. We will analyze the
  construction used in the proof of Theorem~\ref{thm:monoid:to:zvass},
  where reachability in $\V$ is effectively reduced to reachability in
  a $\Z$-VASS $\V' = (d', Q', T')$. Recall that $Q' = (Q \times
  \monoid[\V]) \cup (\overline{Q} \times \monoid[\V]) \cup Q$, and
  thus that the size of $\V'$ depends only on the sizes of $Q$ and
  $\monoid[\V]$.

  It follows from the proof of Theorem~\ref{thm:monoid:to:zvass} and
  Proposition~\ref{prop:theorem_for0} that for every run
  $q'(\vec{u},\vec{0}) \step{*} p(\vec{0},\vec{v})$ in $\V'$ where $q'
  \defeq (q, \mat{I})$, there is a corresponding run $p(\vec{u})
  \step{w} q(\vec{v})$ in $\V$ for some $w \in T^*$ of length $k \ge
  0$. Moreover, the $i^\text{th}$ matrix occuring within the
  control-states of this run are is the form $\mat{A}_i$ where
  $\mat{A}_i = \mat{A}_{i-1} \cdot \mat{B}$ for some $\mat{B} \in
  \monoid[\V]$. Since $\monoid[\V]$ consists of reset matrices, it
  holds that $\mat{A}_0, \mat{A}_1, \mat{A}_2, \ldots, \mat{A}_k$ is
  monotonic, \ie\ if $\mat{A}_{i-1}$ has a $0$ somewhere on its
  diagonal, then $\mat{A}_i$ also contains $0$ in that position. It
  follows that $\mat{A}_0, \mat{A}_1, \ldots, \mat{A}_k$ is made of at
  most $d+1$ distinct matrices.

  To prove the NP upper bound we proceed as follows. We guess at most
  $d+1$ matrices of $\monoid[\V]$ that could appear in sequence
  $\mat{A}_0, \mat{A}_1, \ldots, \mat{A}_k$. We construct the
  $\Z$-VASS $\V'$ as in Theorem~\ref{thm:monoid:to:zvass}, but we
  discard each control-state of $Q'$ containing a matrix not drawn
  from the guessed matrices. Since the constructed $\Z$-VASS is of
  polynomial size, reachability can be verified in NP~\cite{HH14}.
\end{proof}

\begin{rem}
  Observe that the proof of Theorem~\ref{thm:reset:NP} holds for any
  class of affine $\Z$-VASS with a finite monoid such that every path
  of its Cayley graph contains at most polynomially many different
  vertices. For a reset $\Z$-VASS of dimension $d$, the number of
  vertices on every path of the Cayley graph is bounded by $d+1$.
\end{rem}

\section{Hardness results for reachability} \label{sec:lower:bounds}

It is known that the reachability problem for $\Z$-VASS is already
NP-hard~\cite{HH14}, which means that reachability is NP-hard for all
classes of affine $\Z$-VASS. In this section, we show that
PSPACE-hardness holds for some classes, matching the PSPACE upper
bound derived in Section~\ref{sec:upper:bounds}. Moreover, we observe
that reachability is undecidable for transfer + copy $\Z$-VASS.

\begin{thm} \label{thm:pspace:hardness}
  The reachability problem for permutation $\Z$-VASS is
  PSPACE-hard.
\end{thm}

\begin{proof}
  We give a reduction from the membership problem of linear bounded
  automata, which is known to be PSPACE-complete (see,
  \eg,~\cite[Sect.~9.3 and~13]{HU79}). Let $\mathcal{A} = (P, \Sigma,
  \Gamma, \delta, q^\text{ini}, q^\text{acc}, q^\text{rej})$ be a
  linear bounded automaton, where:
  \begin{itemize}
  \item $P$ is the set of states,
  \item $\Sigma \subseteq \Gamma$ is the input alphabet,
  \item $\Gamma$ is the tape alphabet,
  \item $\delta$ is the transition function, and
  \item $q^\text{ini}, q^\text{acc},q^\text{rej}$ are the initial,
    accepting and rejecting states respectively.
  \end{itemize}

  The transition function is a mapping $\delta : P \times \Gamma \to P
  \times \Gamma \times \{\textsc{Left},\textsc{Right}\}$. The intended
  meaning of a transition $\delta(p, a) = (q, b, D)$ is that whenever
  $\mathcal{A}$ is in state $p$ and holds letter $a$ at the current
  position of its tape, then $\mathcal{A}$ overwrites $a$ with $b$ and
  moves to state $q$ and to the next tape position in direction $D$.

  Let us fix a word $w \in \Sigma^*$ of length $n$ that we will check
  for membership. We construct a permutation $\Z$-VASS $\V = (d, Q, T)$
  and configurations $r(\vec{u})$ and $r'(\vec{0})$ such that
  $\mathcal{A}$ accepts $w$ if and only if $r(\vec{u}) \step{*}
  r'(\vec{0})$.
  
  We set $d \defeq n \cdot |\Gamma| + 1$ and associate a counter to
  each position of $w$ and each letter of the tape alphabet $\Gamma$,
  plus one additional counter. For readability, we denote these
  counters respectively as $x_{i, a}$ and $y$, where $i \in [n]$ and
  $a \in \Gamma$. The idea is to maintain, for every $i \in [n]$, a
  single non zero counter among $\{x_{i, a} : a \in \Gamma\}$ in order
  to represent the current letter in the $i^\text{th}$ tape cell of
  $\mathcal{A}$. The initial vector is $\vec{u} \in \{0, 1\}^d$ such
  that $\vec{u}(y) = n$ and $\vec{u}(x_{i, a}) = 1$ if and only if
  $w_i = a$ for every $i \in [n]$ and $a \in \Gamma$.  The invariant
  that will be maintained during all runs is $y = \sum_{i,a}x_{i,a}$.

  The control-states of $\V$ are defined as:
  \begin{align*}
    Q &\defeq \{r_{p, i} : p \in P, i \in [n]\} \cup \{r_{a,i} : a \in \Gamma, i \in [n]\} \cup \{r_\text{acc}\}.
  \end{align*}
  The purpose of each state of the form $r_{p, i}$ is to store the
  current state $p$ and head position $i$ of $\mathcal{A}$. States of the
  form $r_{a, i}$ will be part of a gadget testing whether
  $\mathcal{A}$ is simulated faithfully.

  We associate a transition to every triple $(p, a, i) \in P \times
  \Gamma \times [n]$, which denotes a configuration of $\mathcal{A}$:
  the automaton is in state $p$ in position $i$, where letter $a$ is
  stored. Let us fix a transition $\delta(p, a) = (q, b, D)$; and let
  $j \defeq i + 1$ if $D = \textsc{Right}$, and $j \defeq i - 1$ if $D
  = \textsc{Left}$. For every $i \in [n]$, if $j \in [n]$, then we add
  to $T$ the transition
  \begin{align*}
    (r_{p, i}, \mat{A}, \vec{a}, r_{q, j})
  \end{align*}
  where $\mat{A}$ is a permutation matrix that swaps the values of
  $x_{i, a}$ and $x_{i, b}$; and $\vec{a}$ is the vector whose only nonzero
  components are $\vec{a}(x_{i,b}) = 1$ and $\vec{a}(y) = 1$.
  The transition is
  depicted on the left of Figure~\ref{fig:gadget:step} (for $D =
  \textsc{Right}$). Notice that all transitions maintain the invariant $y = \sum_{i,a}x_{i,a}$.
  
  The purpose of the swap is to simulate the transition of
  $\mathcal{A}$, upon reading $a$ in tape cell $i$ and state $p$, by
  moving the contents from $x_{i, a}$ to $x_{i, b}$. Note that this
  transition may be faulty, \ie\ it can simulate reading letter $a$
  even though tape cell $i$ contains another letter. The purpose of
  the vector $\vec{a}$ is to detect such faulty behaviour: if the cell
  $i$ does not contain $a$, then more than one counter among $\{x_{i,
    a} : a \in \Gamma\}$ will be a nonzero counter.

  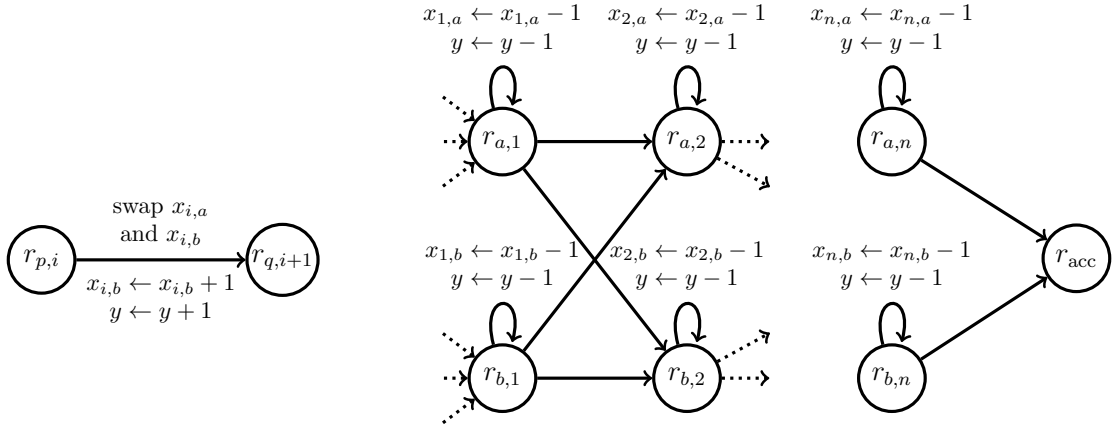
\begin{figure}[h]
    \centering
    \begin{tikzpicture}[->, node distance=2.5cm, auto, very thick, scale=0.9, transform shape, font=\large]
      \tikzset{every state/.style={inner sep=1pt, minimum size=28pt}}
      
      \node[state] (pi) {$r_{p, i}$};
      \node[state] (qj) [right=of pi] {$r_{q, i+1}$};
      
      \node[state] (ra1) [above right = 1cm and 2.5cm of qj] {$r_{a,1}$};
      \node[state, draw=none, minimum size=1pt] (d1) [left=10pt of ra1] {};
      \node[state, draw=none, minimum size=1pt] (d2) [left=10pt of ra1, yshift=-20pt] {};
      \node[state, draw=none, minimum size=1pt] (d3) [left=10pt of ra1, yshift=20pt] {};

      \node[state] (rb1) [below right = 1cm and 2.5cm of qj] {$r_{b,1}$};
      \node[state, draw=none, minimum size=1pt] (e1) [left=10pt of rb1] {};
      \node[state, draw=none, minimum size=1pt] (e2) [left=10pt of rb1, yshift=-20pt] {};
      \node[state, draw=none, minimum size=1pt] (e3) [left=10pt of rb1, yshift=20pt] {};
      
      \node[state] (ra2) [right = 1.7cm of ra1] {$r_{a,2}$};
      \node[state] (rb2) [right = 1.7cm of rb1] {$r_{b,2}$};
      
      \node[state] (ran) [right = 2cm of ra2] {$r_{a,n}$};
      \node[state, draw=none, minimum size=1pt] (f1) [right=20pt of ra2] {};
      \node[state, draw=none, minimum size=1pt] (f2) [right=20pt of ra2, yshift=-20pt] {};

      \node[state] (rbn) [right = 2cm of rb2] {$r_{b,n}$};
      \node[state, draw=none, minimum size=1pt] (g1) [right=20pt of rb2] {};
      \node[state, draw=none, minimum size=1pt] (g2) [right=20pt of rb2, yshift=20pt] {};
      
      \node[state] (racc) [below right = 1cm and 2cm of ran] {$r_\text{acc}$};
      
      \path[->, font=\small]
      (pi) edge [] node[above] {\begin{tabular}{c}swap $x_{i, a}$ \\ and
          $x_{i, b}$\end{tabular}} node[below] {\begin{tabular}{c} $x_{i, b} \leftarrow x_{i,b} + 1$ \\
          $y \leftarrow y + 1$\end{tabular}} (qj)

      (ra1) edge[loop above] node {
        \begin{tabular}{c}
          $x_{1, a} \leftarrow x_{1, a} - 1$ \\
          $y \leftarrow y - 1$
        \end{tabular}
      } (ra1)
      
            (rb1) edge[loop above] node {
        \begin{tabular}{c}
          $x_{1, b} \leftarrow x_{1, b} - 1$ \\
          $y \leftarrow y - 1$
        \end{tabular}
      } (rb1)
      
      (ra2) edge[loop above] node {
        \begin{tabular}{c}
          $x_{2, a} \leftarrow x_{2, a} - 1$ \\
          $y \leftarrow y - 1$
        \end{tabular}
      } (ra2)
      
            (rb2) edge[loop above] node {
        \begin{tabular}{c}
          $x_{2, b} \leftarrow x_{2, b} - 1$ \\
          $y \leftarrow y - 1$
        \end{tabular}
      } (rb2)
      
            (ran) edge[loop above] node {
        \begin{tabular}{c}
          $x_{n, a} \leftarrow x_{n, a} - 1$ \\
          $y \leftarrow y - 1$
        \end{tabular}
      } (ran)
      
            (rbn) edge[loop above] node {
        \begin{tabular}{c}
          $x_{n, b} \leftarrow x_{n, b} - 1$ \\
          $y \leftarrow y - 1$
        \end{tabular}
      } (rbn)
      
      (ra1) edge (ra2)
      (ra1) edge (rb2)
      (rb1) edge (rb2)
      (rb1) edge (ra2)
      (ran) edge (racc)
      (rbn) edge (racc)
      ;

      \path[->, dotted]
      (d1) edge node {} (ra1)
      (d2) edge node {} (ra1)
      (d3) edge node {} (ra1)
      (e1) edge node {} (rb1)
      (e2) edge node {} (rb1)
      (e3) edge node {} (rb1)
      (ra2) edge (f1)
      (ra2) edge (f2)
      (rb2) edge (g1)
      (rb2) edge (g2)
      ;
    \end{tikzpicture}
    \caption{\emph{Left}: transitions of $\V$ simulating transition
      $\delta(p, a) = (q, b, \textsc{Right})$ of
      $\mathcal{A}$. \emph{Right}: gadget verifying whether the
      accepting state has been reached with no error during the
      simulation. For readability, we assume $\Gamma = \{a, b\}$ in
      the right gadget.}\label{fig:gadget:step}
  \end{figure}

  Recall that $y = \sum_{i,a}x_{i,a}$. We conclude that $\mathcal{A}$
  accepts $w$ if and only if there exist $j \in [n]$, $\vec{u}' \in
  \N^d$ and $a_1, a_2, \ldots, a_n \in \Gamma$ such that
  $$r_{q^\text{ini},1}(\vec{u}) \step{*} r_{q^\text{acc},j}(\vec{u}')
  \text{ and } \vec{u}'(y) = \sum_{i \in [d]} \vec{u}'(x_{i,a_i}).$$
  
  To test whether such index $j$, vector $\vec{u}'$ and letters $a_1,
  a_2, \ldots, a_n$ exist, we add some transitions to $T$ as
  illustrated on the right of Figure~\ref{fig:gadget:step}.  For every
  $i \in [n]$ and every $a \in \Gamma$, we add to $T$ the transitions
  $(r_{q^\text{acc},i}, \mat{I}, \vec{0}, r_{a,1})$. For every $i \in
  [n]$ and $a \in \Gamma$, we add to $T$ the transitions $(r_{a,i},
  \mat{I}, \vec{b}, r_{a,i})$ where $\vec{b}$ is the vector whose only
  non zero components are $\vec{b}(x_{i, a}) = \vec{b}(y) =
  -1$. Moreover, if $i < n$, then for every $a,b \in \Gamma$ we also
  add transitions $(r_{a,i}, \mat{I}, \vec{0}, r_{b,i+1})$. Finally,
  for all $a \in \Gamma$, we also add transitions $(r_{a,n}, \mat{I},
  \vec{0}, r_\text{acc})$. The purpose of these transitions is to
  guess for each $i$ some letter $a_i$ and simultaneously decrease
  $x_{i,a_i}$ and $y$. We do this for each $i$ starting from $1$ to
  $n$ and in the end we move to the state $r_\text{acc}$.  We conclude
  that $\mathcal{A}$ accepts $w$ if and only if $r_{q^\text{ini},
    1}(\vec{u}) \step{*} r_\text{acc}(\vec{0})$ in $\V$.
\end{proof}

\begin{cor}
  The reachability problem is PSPACE-complete for permutation
  $\Z$-VASS, transfer $\Z$-VASS and copy $\Z$-VASS.
\end{cor}

\begin{proof}
  PSPACE-hardness for permutation $\Z$-VASS was shown in
  Theorem~\ref{thm:pspace:hardness}, and the upper bound for transfer
  $\Z$-VASS and copy $\Z$-VASS follows from
  Theorem~\ref{thm:pspace:upper}. It remains to observe that
  permutation matrices are also transfer and copy matrices.
\end{proof}

\begin{prop}\label{proposition:undecidable}
  The reachability problem for transfer + copy $\Z$-VASS is
  undecidable, even when restricted to three counters.
\end{prop}

\begin{proof}
  Reichert~\cite{Rei15} gives a reduction from the Post correspondence
  problem over the alphabet $\{0, 1\}$ to reachability in affine
  $\Z$-VASS with two counters.
  \begin{figure}[h]
    \centering
    \begin{tikzpicture}[->, node distance=3cm, auto, very thick, scale=0.9, transform shape, font=\large]
      \tikzset{every state/.style={inner sep=1pt, minimum size=25pt}}

      \node[state] (pa) {$p$};
      \node[state] (qa) [right of=pa] {$q$};

      \node[state] (pb) [right=1.5cm of qa] {$p$};
      \node[state] (r) [right=of pb] {};
      \node[state] (qb) [right=of r] {$q$};

      \path[->]
      (pa) edge node {$\mat{D_1}, \begin{pmatrix}b_1 \\ b_2\end{pmatrix}$} (qa)

      (pb) edge node {$
        \begin{pmatrix}1 & 0 & 0 \\ 0 & 1 & 0 \\ 1 & 0 & 0\end{pmatrix}
        $, $\vec{0}$} (r)
      (r) edge node {$
        \begin{pmatrix}1 & 0 & 1 \\ 0 & 1 & 0 \\ 0 & 0 & 0\end{pmatrix}
        $, $\begin{pmatrix}b_1 \\ b_2 \\ 0\end{pmatrix}$} (qb)
      ;
    \end{tikzpicture}
    \caption{Gadget (on the right) made of copy and transfer
      transitions simulating the doubling transition on the
      left.}\label{fig:doubling:gadget}
  \end{figure}
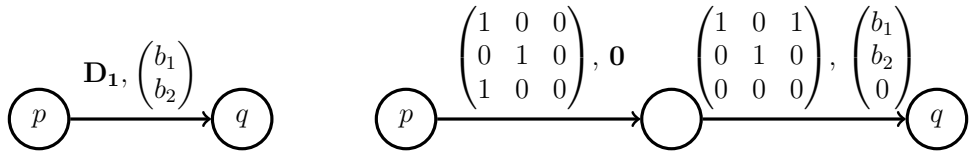
  The trick of the reduction is to represent two binary sequences as
  the natural numbers the sequences encode, one in each counter. If we
  add an artificial $1$ at the beginning of the two binary sequences,
  then these sequences are uniquely determined by their numerical
  values. We only need to be able to double the counter values, which
  corresponds to shifting the sequences. This can be achieved using the
  following matrices:
  \begin{align*}
    \mat{D_1} &\defeq \begin{pmatrix}2 & 0 \\ 0 & 1\end{pmatrix} \text{ and }
    \mat{D_2}  \defeq \begin{pmatrix}1 & 0 \\ 0 & 2\end{pmatrix}.
  \end{align*}  

  The only matrices used in the construction of Reichert are
  $\mat{I}$, $\mat{D_1}$ and $\mat{D_2}$. The last two matrices can be
  simulated by a gadget made of copy and transfer matrices and by
  introducing a third counter. This gadget is depicted in
  Figure~\ref{fig:doubling:gadget} for the case of matrix
  $\mat{D_1}$. The other gadget is symmetric. Note that if a
  run enters control-state $p$ of the gadget with vector $(x, y,
  0)$, then it leaves control-state $q$ in vector $(2x + b_1, y + b_2,
  0)$ as required.
\end{proof}

\begin{rem}
  The coverability problem for nonnegative affine VASS is known to be
  decidable in Ackermann time~\cite{FFSS11}. Recall that coverability
  and reachability are inter-reducible for affine $\Z$-VASS. Thus,
  Proposition~\ref{proposition:undecidable} gives an example of a
  decision problem, namely coverability, which is more difficult for
  affine $\Z$-VASS than for affine VASS.
\end{rem}

\section{Reachability beyond finite monoids}
\label{sec:inf:monoid}

Thus far, we have shown, on the one hand, that reachability is decidable
for affine $\Z$-VASS with the finite-monoid property, and, on the other
hand, that reachability is undecidable for arbitrary affine
$\Z$-VASS. This raises the question of whether there is a decidability
dichotomy between classes of finite and infinite monoids, \ie\ whether
reachability is undecidable for \emph{every} class of infinite
monoids. In this section, we show that this is not the case: we
exhibit a non-trivial class of infinite monoids for which affine $\Z$-VASS
reachability is \emph{decidable}. In other words, the top rectangular
region of Figure~\ref{fig:classification} is \emph{not} equal to the
red ellipse, which answers a question we left open
in~\cite{BlondinHM18}. The class of affine $\Z$-VASS will have a
particular shape, namely, the matrix monoids have a single
generator. More formally, we say that a class of matrices $C =
\bigcup_{d \geq 1} \class_d$ is \emph{monogenic} if each monoid
$\class_d$ is generated by a single matrix. In the second part of this
section we prove that reachability is in general undecidable for
monogenic classes.

\subsection{Decidability for a class of affine $\Z$-VASS with infinite monoids}\label{sec:inf:monoids}

Let $\class_d$ be the monoid generated by the (nonnegative) matrix
$\matone_d \in \N^{d \times d}$ whose entries are all equal to
$1$. Clearly, $\class_d$ is infinite for every $d \ge 2$ since
$(\matone_d)^n$ is the matrix whose entries are all equal to
$d^n$. Let $\classone = \bigcup_{d \geq 1} \class_d$. The rest of this
section is devoted to proving the following theorem:

\begin{thm}\label{thm:infinite_monoid}
  The reachability problem $\reach[\classone]$ is decidable.
\end{thm}

Let $\V = (d, Q, T)$ be an affine $\Z$-VASS belonging to
$\classone$. We will simply write $\matone$ instead of $\matone_d$ as
$d$ is implicit from the dimension of $\V$. Observe that we can assume
w.l.o.g.\ that for every transition $(p, \mat{A}, \vec{b}, q) \in T$
either $\mat{A} = \mat{I}$ or $\vec{b} = \vec{0}$, \ie\ each
transition either performs a transformation of the form $\vec{x}
\leftarrow \vec{x} + \vec{b}$ or $\vec{x} \leftarrow \mat{A} \cdot
\vec{x}$. Indeed, by adding an extra state $r$, we can always split
such a transition into two transitions $(p, \mat{A}, \vec{0}, r)$ and
$(r, \mat{I}, \vec{b}, q)$. We can further assume w.l.o.g.\ that
$\mat{I}$ and $\matone$ are the only matrices occurring in
$\V$. Indeed, if $T$ contains a transition $t = (p, \mat{A}, \vec{b},
q)$ where $\mat{A} \not\in \{\mat{I}, \matone\}$, then $\mat{A} =
\matone^n$ for some $n \geq 2$ and $\vec{b} = \vec{0}$. Thus, we can
simply replace $t$ by a sequence of transitions $t_1, t_2, \ldots,
t_n$ leading from $p$ to $q$ and such that $\tmat[t_i] = \matone$ and
$\tvec[t_i] = \vec{0}$ for every $i \in [n]$.

Let $\Tid$ and $\Tone$ denote the (maximal) subsets of $T$ of
transitions with matrix $\mat{I}$ and $\matone$ respectively. Note
that $\Tid$ and $\Tone$ form a partition of $T$. We
will write $\step{S}$ and $\step{S^*}$ to denote respectively the
restriction of $\step{}$ and $\step{*}$ to transitions of a set
$S$. We give a simple characterization of reachability in $\V$:

\begin{prop}\label{prop:two_parts}
  For all configurations $p(\vec{u})$ and $q(\vec{v})$ of $\V$,
  $p(\vec{u}) \step{*} q(\vec{v})$ if and only if:
  \begin{enumerate}
  \item\label{c:1} $p(\vec{u}) \step{T_{\mat{I}}^*} q(\vec{v})$; or

  \item\label{c:2} $p(\vec{u}) \step{*} r'(\vec{w}) \step{\Tone}
    r(\matone \cdot \vec{w}) \step{\Tid^*} q(\vec{v})$ for some $r,
    r' \in Q$ and $\vec{w} \in \Z^d$.
  \end{enumerate}
\end{prop}

\begin{proof}
  $\Leftarrow$) Immediate.

  $\Rightarrow$) Assume that $p(\vec{u}) \step{w} q(\vec{v})$ for some
  $w \in T^*$. If $w$ does not contain any transition from $\Tone$,
  then~\eqref{c:1} holds and we are done. Thus, suppose that $w$
  contains at least one transition from $\Tone$. Let $t \in \Tone$ be
  the last such transition occurring in $w$. Recall that, by
  assumption, $\tmat[t] = \matone$ and $\tvec[t] =
  \vec{0}$. Therefore, we are done since there exist $r, r' \in Q$ and
  $\vec{w} \in \Z^d$ such that
  \begin{align*}
    p(\vec{u}) \step{*} r'(\vec{w}) \step{t} r(\matone \cdot \vec{w})
    \step{\Tid^*} q(\vec{v}). &\qedhere
  \end{align*}
\end{proof}

In order to prove that $\reach[\classone]$ is decidable, it suffices
to show that there exist procedures to decide the two conditions of
Proposition~\ref{prop:two_parts}. Testing condition~\eqref{c:1}
amounts to $\Z$-VASS reachability, which belongs to
NP~\cite{HH14}. Indeed, any run restricted to $\Tid$ is a run of the
$\Z$-VASS induced by $\Tid$. Thus, in the rest of the proof, we focus
on showing how to test condition~\eqref{c:2}.

For this purpose, let us introduce an auxiliary model. An \emph{affine
  one-counter $\Z$-net} is a pair $(P, U)$ where
\begin{itemize}
\item $P$ is a finite set of \emph{states}, and

\item $U \subseteq Q \times \{+, \cdot\} \times \Z \times Q$ is a
  finite set of \emph{transitions}.
\end{itemize}

Furthermore, for every transition $t = (p, \circledast, c, q)$, we
write $p(n) \step{t} q(m)$ if $m = n\circledast c$. The notions of
runs and reachability are defined accordingly as for affine
$\Z$-VASS. These machines are a special case of one-counter register
machines with polynomial updates whose reachability problem belongs to
PSPACE~\cite{FinkelGH13}, \ie\ we only allow the counter to be
multiplied or incremented by constants, whereas the model
of~\cite{FinkelGH13} allows to update the counter by a polynomial such
as $x^2$ or $x^3 - x + 1$.

For every $\vec{v} \in \Z^d$, let $$\delta(\vec{v}) \defeq
\sum_{i=1}^d \vec{v}(i).$$ 

Consider the transitions $T$ in the affine $\Z$-VASS $\V$. For every transition $t \in T$, let $\overline{t}$ be defined as:
$$
\overline{t} \defeq
\begin{cases}
  (p, +, \delta(\tvec[t]), q) & \text{if}\ t \in \Tid, \\
  (p, \cdot, d, q)            & \text{if}\ t \in \Tone,
\end{cases}
$$ where $p = \src{t}$ and $q = \tgt{t}$.

Let $\W = (Q, \overline{T})$ be the affine one-counter $\Z$-net
obtained from $\V = (d, Q, T)$ by keeping the same states and taking
$\overline{T} \defeq \{\overline{t} : t \in T\}$. We write
$\overline{w} \in \overline{T}^*$ to denote the (unique) sequence of transitions in
$\W$ corresponding to the sequence $w \in T^*$ of $\V$. Let us observe
the following correspondence between $\V$ and $\W$:

\begin{lem}\label{lemma:V_to_W}
  For every $p, q \in Q$, $\vec{u} \in \Z^d$, $m \in \Z$ and $w \in
  T^*$, we have $p(\delta(\vec{u})) \step{\overline{w}} q(m)$ in $\W$
  if and only if $p(\vec{u}) \step{w} q(\vec{v})$ in $\V$ for some
  $\vec{v} \in \Z^d$ such that $\delta(\vec{v}) = m$.
\end{lem}

\begin{proof}
  The claim follows from a simple induction on $|w|$.
\end{proof}

We may now prove Theorem~\ref{thm:infinite_monoid}.

\begin{proof}[Proof of Theorem~\ref{thm:infinite_monoid}]
Recall that it
suffices to show how to decide condition~\eqref{c:2} of
Proposition~\ref{prop:two_parts}. By definition of $\matone$, this
condition is equivalent to determining whether there exist $r \in Q$
and $n \in \Z$ such that
\begin{align*}
  p(\vec{u}) \step{T^* \cdot \Tone} r(n, n, \ldots, n) \step{\Tid^*}
  q(\vec{v}).\label{eq:overall:cond}
\end{align*}

Let $S = \{m \in \Z : \exists n \in \Z\ (m = d \cdot n) \land
\bigvee_{r \in Q} r(n, n, \ldots, n) \step{\Tid^*} q(\vec{v})\}$. As
we mentioned earlier, $\V$ can be seen as a standard $\Z$-VASS when
restricted to $\Tid$. Since the reachability relation of any $\Z$-VASS
is effectively semilinear~\cite{HH14}, the set $S$ is also effectively
semilinear.

\begin{figure}[!h]
  \definecolor{colnew}{RGB}{25, 158, 213}
  \centering
  \begin{tikzpicture}[->, node distance=1.25cm, auto, very thick, scale=0.9, transform shape, font=\large]
    \tikzset{every state/.style={inner sep=1pt, minimum size=25pt}}
 
    \node[state] (p) {$p$};
    \node[]      (d1) [below=0.25cm of p]  {$\vdots$};
    \node[state] (i1) [below=0.25cm of d1] {};
    \node[]      (d2) [below=0.25cm of i1] {$\cdots$};
    \node[state] (i2) [below=0.25cm of d2] {};
    \node[state] (r)  [left= of i2]        {$r$};

    \node[state, fill=colnew] (a)  [right= of i2]      {};
    \node[state, fill=colnew] (b)  [below=4cm of a] {};
    \node[state, fill=colnew] (c)  [right=4cm of a] {$r'$};

    \path[->, font=\large]
    (i1) edge[bend right] node[swap, xshift=5pt] {$\overline{\Tone}$} (r)
    (i2) edge[]           node[swap] {$\overline{\Tone}$} (r)
    ;

    \path[->, font=\large, colnew, text=black]
    (i1) edge[bend left]  node[xshift=-5pt] {$\cdot d$} (a)
    (i2) edge[]           node[] {$\cdot d$} (a)

    (a) edge[out=45,  in=135] node[] {$-f_1$} (c)
    (a) edge[out=15,  in=165] node[] {$-f_2$} (c)
    (a) edge[out=-45, in=-135] node[] {$-f_k$} (c)

    (a) edge[out=-135,  in=135] node[swap, xshift=2pt] {$-b_1$} (b)
    (a) edge[out=-105,  in=105] node[swap, xshift=2pt] {$-b_2$} (b)
    (a) edge[out=-45, in=45] node[swap, xshift=2pt] {$-b_\ell$} (b)

    (b) edge[bend right] node[swap] {$+0$} (c)
    (b) edge[loop left] node[] {$-a$} ()
    ;

    \path[->, dotted, colnew]
    (a) edge[out=-15, in=-165] node[] {} (c)
    (a) edge[out=-75, in=75]   node[] {} (b)
    ;
    
  \end{tikzpicture}
  \caption{Affine one-counter $\Z$-net $\W$ extended with a gadget
    subtracting a number of $S$ from some transition of
    $\overline{\Tone}$ leading to $r$. The gadget is depicted in
    colour. Transitions connecting $\W$ to the gadget are labeled with
    ``$\cdot d$'' as this is the effect of every transition of
    $\overline{\Tone}$.}\label{fig:transformation}
\end{figure}
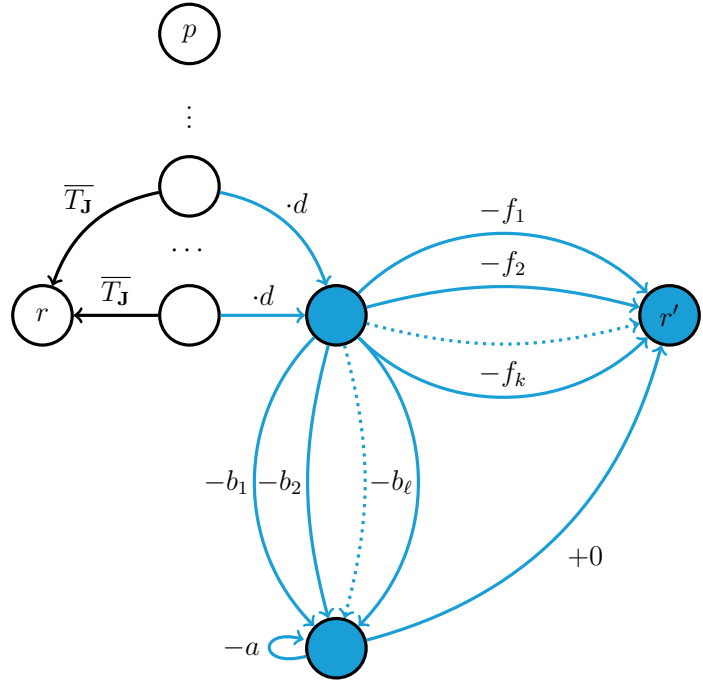

By Lemma~\ref{lemma:V_to_W}, we have $p(\vec{u}) \step{T^* \cdot
  \Tone} r(n, n, \ldots, n)$ if and only if $p(\delta(\vec{u}))
\step{\overline{T}^* \cdot \overline{\Tone}} r(d \cdot n)$. Indeed,
the direction ($\Rightarrow$) is immediate. To prove the implication ($\Leftarrow$) suppose $p(\delta(\vec{u}))
\step{\overline{T}^* \cdot \overline{\Tone}} r(d \cdot n)$.
By Lemma~\ref{lemma:V_to_W} we have $p(\vec{u}) \step{T^* \cdot
  \Tone} r(\vec{v})$ such that $\delta(\vec{v}) = d \cdot n$. By definition the last transition is $r'(\vec{w}) \step{t} r(\matone \cdot \vec{w})$, where $\vec{v} = \matone \cdot \vec{w}$.
By definition of $\matone$: $\vec{v} = (\delta(\vec{w}), \ldots, \delta(\vec{w}))$. Since $\delta(\vec{v}) = d \cdot n$ we get $\vec{v} = (n, n, \ldots, n)$.

Thus, it suffices to test whether $p(\delta(\vec{u}))
\step{\overline{T}^* \cdot \overline{\Tone}} r(m)$ in $\W$ for some $m
= d \cdot n \in S$. This can be achieved by extending $\W$ with a
gadget that non deterministically subtracts some element of $S$ after
executing a transition from $\overline{\Tone}$. More precisely, since
$S$ is an (effectively) semilinear set of integers, it is also
(effectively) ultimately periodic. Thus, it is possible to obtain a
description of $S = F \cup B + a \cdot \N$ where $F = \{f_1, f_2,
\ldots, f_k\}$ and $B = \{b_1, b_2, \ldots, b_\ell\}$ are finite
subsets of $\Z$. We extend $\W$ with the gadget depicted in
Figure~\ref{fig:transformation}. More precisely, for every transition
$t \in \overline{\Tone}$ leading to $r$, we add a new transition
leading to a gadget that either subtracts some number from $F$ or
some number from $B + a \cdot \N$. Note that the gadget is not
``attached'' directly to $r$ as we must ensure that $r$ is entered by
a transition of $\overline{\Tone}$. Hence, testing
whether $$p(\delta(\vec{u})) \step{\overline{T}^* \cdot
  \overline{\Tone}} r(m) \text{ in } \W \text{ for some } m \in S$$
amounts to testing whether $p(\delta(\vec{u})) \step{*} r'(0)$ in the
new net. Since the latter can be done in polynomial
space~\cite{FinkelGH13}, we are done.
\end{proof}

\subsection{Undecidability for monogenic classes}

In contrast with the previous result, we prove that decidability is
undecidable in general for monogenic classes:

\begin{thm}\label{thm:monogenic}
  Reachability for monogenic affine $\Z$-VASS is
  undecidable. Moreover, there exists a fixed monogenic affine
  $\Z$-VASS for which deciding reachability is undecidable.
\end{thm}

We show the first part of Theorem~\ref{thm:monogenic} by giving a
reduction from the problem of determining whether a given Diophantine
equation has a solution over the natural numbers, which is well-known
to be undecidable. The second part of Theorem~\ref{thm:monogenic}
follows as a corollary. Indeed, by Matiyasevich's theorem, Diophantine
sets correspond to recursively enumerable sets. In particular, there
exists a polynomial $P$ such that
$$x \in \N \text{ is the encoding of a halting Turing machine} \iff
\exists \vec{y} : P(x, \vec{y}) = 0.$$ The forthcoming construction
will yield a monogenic affine $\Z$-VASS that can test ``$\exists
\vec{y} : P(x, \vec{y}) = 0$'' by nondeterministically guessing
$\vec{y}$ and testing $P(x, \vec{y}) = 0$. Hence, reachability cannot
be decided for this monogenic affine $\Z$-VASS as the above language
is undecidable.

Let us show the first part of Theorem~\ref{thm:monogenic}. Let $x_1,
x_2, \ldots, x_k$ be variables of a given polynomial $P(x_1, x_2,
\ldots, x_k)$. We will construct an instance of the reachability
problem, for a monogenic affine $\Z$-VASS $\V$, such that reachability
holds if and only if $P(x_1, x_2, \ldots, x_k) = 0$ has a solution
over $\N^k$.

The affine $\Z$-VASS will be described using the syntax of counter
programs; see~\cite{Esparza98,CzerwinskiLLLM19}, where a similar
syntax was used to present the VASS model. We will make use of two
instructions: \testz{\vr{x}}\ and \textbf{loop}. The former checks
whether counter $\vr{x}$ has value $0$, and the latter repeats a block
of instructions an arbitrary number of
times. Figure~\ref{fig:programs} gives an example of such a program
together with its translation as an affine $\Z$-VASS.

\begin{figure}[h!]
  \centering
  \begin{subfigure}{.3\textwidth}
    \centering
    \begin{algorithmic}
      \Loop
      \State \add{\vr{x}}{1}
      \State \add{\vr{y}}{4}
      \EndLoop
      \Loop
      \State \sub{\vr{y}}{3\vr{x} + 1}
      \EndLoop
      \State \testz{\vr{y}}
    \end{algorithmic}
  \end{subfigure}
  \begin{subfigure}{.5\textwidth}
    \centering
    \begin{tikzpicture}[->, node distance=3cm, auto, very thick, scale=0.9, transform shape, font=\large]
      \tikzset{every state/.style={inner sep=1pt, minimum size=25pt}}
      
      \node[state] (p) {};
      \node[state, right = 3cm of p] (q) {};

      \path
      (p) edge[loop above] node {
          $\begin{pmatrix}1 & 0 \\ 0 & 1\end{pmatrix}$,
          $\begin{pmatrix}1 \\ 4\end{pmatrix}$}
      (p)
      
      (q) edge[loop above] node {
        $\begin{pmatrix}1 & 0 \\ -3 & 1\end{pmatrix}$,
          $\begin{pmatrix}0 \\ 1\end{pmatrix}$}
      (q)

      (p) edge[below] node {
        $\begin{pmatrix}1 & 0 \\ 0 & 1\end{pmatrix}$,
          $\begin{pmatrix}0 \\ 0\end{pmatrix}$
      }
      (q)
      ;
    \end{tikzpicture}
  \end{subfigure}
  \caption{\emph{Left}: an example program $\mathcal{P}$ using
    instructions \textbf{zero?}\ and \textbf{loop}. \emph{Right}: an
    affine $\Z$-VASS $\V$ equivalent to $\mathcal{P}$, where its first
    and second components correspond to counters $\vr{x}$ and $\vr{y}$
    respectively. The program loops are simulated by loops within the
    control structure of $\V$. Note that whenever $\mathcal{P}$ only
    adds and subtracts constants from counters, the associated matrix
    is the identity. Since the only way $\V$ can test whether a
    counter equals $0$ is at the end of the program via a reachability
    query, instruction \testz{\vr{y}}\ merely emphasizes that counter
    $\vr{y}$ will never be used again.}\label{fig:programs}
\end{figure}
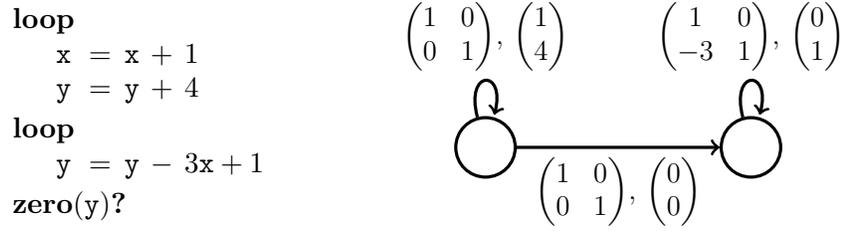

\parag{Macros.} Before describing the reduction, let us introduce
helpful macros. First, we define macros ``\transfer{\vr{x}}{\vr{y}}''
and ``\remove{\vr{x}}{\vr{y}}''. The former computes $\vr{y} = \vr{y}
+ \vr{x}$ and $\vr{x} = 0$, and the latter computes $\vr{y} = \vr{y} -
\vr{x}$ and $\vr{x} = 0$. Both macros work under the assumption that
$\vr{x}$ is initially non negative. These macros are implemented as
follows:
\begin{center}
  \alggap
  \begin{minipage}{0.49\textwidth}
    \begin{algorithmic}
      \Function{\transfer{\vr{x}}{\vr{y}}:}{}
      \COMMENT{pre-cond.: $\vr{x} \geq 0$}
      \Loop
      \State \sub{\vr{x}}{1}
      \State \add{\vr{y}}{1}
      \EndLoop
      \State \testz{\vr{x}}
      \EndFunction
    \end{algorithmic}
  \end{minipage}
  \begin{minipage}{0.49\textwidth}
    \begin{algorithmic}
      \Function{\remove{\vr{x}}{\vr{y}}:}{}
      \COMMENT{pre-cond.: $\vr{x} \geq 0$}
      \Loop
      \State \sub{\vr{x}}{1}
      \State \sub{\vr{y}}{1}
      \EndLoop
      \State \testz{\vr{x}}
      \EndFunction
    \end{algorithmic}
  \end{minipage}
  \alggap
\end{center}

We define another macro ``\square{\vr{s}}{\vr{t}}'' for squaring the
contents of a counter. More precisely, it computes $\vr{s} = \vr{t}^2$
and $\vr{t} = 0$. This macro is implemented as follows:
\begin{center}
  \alggap
  \begin{minipage}{0.98\textwidth}
    \begin{algorithmic}[1]
      \Function{\square{\vr{s}}{\vr{t}}:}{}
      \COMMENT{pre-condition: $\vr{t} = \vr{x} = \vr{y} = \vr{z} = 0$}

      \State \transfer{\vr{s}}{\vr{x}} \label{l:tr}
      \Loop \label{l:s}
      \State \sub{\vr{x}}{1}
      \State \add{\vr{y}}{1}
      \State \add{\vr{z}}{2\vr{y} + 1} \label{l:e}
      \EndLoop
      \State \testz{\vr{x}}
      \State \transfer{\vr{z}}{\vr{t}}
      \COMMENT{$\vr{t} = \vr{s}^2$, $\vr{y} = \vr{s}$ and $\vr{x} =
        \vr{z} = 0$}
      
      \State \remove{\vr{s}}{\vr{y}} \label{l:rem}
      \State \testz{\vr{y}} \COMMENT{$\vr{y} = 0$}

      \EndFunction
    \end{algorithmic}
  \end{minipage}
  \alggap
\end{center}
The above program starts with $\vr{t}$ and its auxiliary counters set
to $0$, and ends with $\vr{s}$ and the auxiliary counters set to
$0$. Its correctness follows by observing that $(n+1)^2 = n^2 + (2n +
1)$.

We introduce one last macro ``\multip{\vr{s}}{\vr{s'}}{\vr{t}}'' for
multiplication. More precisely, it computes $\vr{t} = \vr{s} \cdot
\vr{s}'$ and $\vr{s} = \vr{s}' = 0$. Its implementation exploits the
fact that $2mn = (m + n)^2 - m^2 - n^2$:
\begin{center}
  \alggap
  \begin{minipage}{0.98\textwidth}
    \begin{algorithmic}
      \Function{\multip{\vr{s}}{\vr{s}'}{\vr{t}}:}{}
      \COMMENT{pre-condition: $\vr{t} = \vr{x} = \vr{y} = \vr{z} =
        \vr{z}' = 0$}

      \State \square{\vr{s}}{\vr{x}}
      \COMMENT{$\vr{x} = \vr{s}^2$}
      \State \square{\vr{s}'}{\vr{y}}
      \COMMENT{$\vr{y} = (\vr{s}')^2$}
      \State \transfer{\vr{s}}{\vr{z}'}
      \State \transfer{\vr{s}'}{\vr{z}'}
      \State \square{\vr{z}'}{\vr{z}}
      \COMMENT{$\vr{z} = (\vr{s} + \vr{s}')^2$, $\vr{z}' = 0$}
      \State \remove{\vr{x}}{\vr{z}}
      \State \remove{\vr{y}}{\vr{z}}
      \COMMENT{$\vr{z} = 2 \cdot \vr{s} \cdot \vr{s}'$, $\vr{x} = \vr{y} = 0$}
      \Loop
      \State \sub{\vr{z}}{2}
      \State \add{\vr{t}}{1}
      \EndLoop
      \State \testz{\vr{z}}
      \COMMENT{$\vr{t} = \vr{s} \cdot \vr{s}'$, $\vr{z} = 0$}

      \EndFunction
    \end{algorithmic}
  \end{minipage}
  \alggap
\end{center}
The above program starts with $\vr{t}$ and its auxiliary counters set
to $0$, and ends with $\vr{s}$, $\vr{s}'$ and its auxiliary counters
set to $0$. Note that a macro ``\multip{\vr{s}}{c}{\vr{t}}'' for
multiplying by a constant $c$ can be achieved by a simpler program:
\begin{center}
  \alggap
  \begin{minipage}{0.98\textwidth}
    \begin{algorithmic}
      \Function{\multip{\vr{s}}{c}{\vr{t}}:}{}
      \COMMENT{pre-condition: $\vr{t} =  0$}

      \Loop
      \State \sub{\vr{s}}{1}
      \State \add{\vr{t}}{c}
      \EndLoop
      \State \testz{\vr{s}}

      \EndFunction
    \end{algorithmic}
  \end{minipage}
  \alggap
\end{center}

Although these programs can be implemented rather straightforwardly by
an affine $\Z$-VASS, two remarks are in order:
\begin{itemize}
\item Affine $\Z$-VASS do not have any native operation for testing a
  counter for zero. However, a counter $\vr{x}$ can be tested
  \emph{once} via a reachability query, provided that $\vr{x}$ is left
  untouched after instruction \testz{x}\ has been
  invoked. Consequently, a \emph{constant} number of zero-tests can be
  performed on a counter, provided its initial contents has been
  duplicated;
  
\item Instruction ``\transfer{\vr{s}}{\vr{t}}'' at line~\ref{l:tr} of
  macro ``\square{\vr{s}}{\vr{t}}'' destroys the contents of $\vr{s}$
  which is later needed at line~\ref{l:rem}. As for zero-tests, this
  is not an issue provided that some counter holds a copy of
  $\vr{s}$. Thus, only a \emph{constant} number of squaring, and hence
  of multiplications, can be performed from a given counter.
\end{itemize}

\parag{The construction.} Let us now describe the reachability
instance. The initial vector is $\vec{0}$, which corresponds to having
all counters set to $0$ at the start of the program. The target vector
is also $\vec{0}$, which corresponds to performing zero tests on all
counters.

The program starts by performing a sequence of loops that guess a
valuation $\vec{x}$ for which $P(\vec{x}) = 0$ is to be tested. More
precisely, a value is nondeterministically picked for each variable
$x_i$ and stored in counters $\vr{x}_i^1, \vr{x}_i^2, \ldots,
\vr{x}_i^{n_i}$. The reason for having $n_i$ copies of the value is to
address the two issues mentioned earlier concerning zero-tests and
reusing counters within macros. The precise number of copies, $n_i$,
will be determined later. The fragment of code achieving the
initialization is as follows:
\begin{center}
  \alggap
  \begin{minipage}{0.98\textwidth}
    \begin{algorithmic}
      \Loop
      \State \add{\vr{x}_1^1}{1}; \quad \add{\vr{x}_1^2}{1}; \quad
      $\cdots$ \quad \add{\vr{x}_1^{n_1}}{1}
      \EndLoop
      
      \Loop
      \State \add{\vr{x}_2^1}{1}; \quad \add{\vr{x}_2^2}{1}; \quad
      $\cdots$ \quad \add{\vr{x}_2^{n_2}}{1}
      \EndLoop
      
      \State \quad\vdots
      \vspace{1pt}
      
      \Loop
      \State \add{\vr{x}_k^1}{1}; \quad \add{\vr{x}_k^2}{1}; \quad
      $\cdots$ \quad \add{\vr{x}_k^{n_k}}{1}
      \EndLoop
    \end{algorithmic}
  \end{minipage}
  \alggap
\end{center}

After the initialization, we compute the value of each monomial
occurring within polynomial $P(x_1, x_2, \ldots, x_k)$. This can be
achieved using counters $\vr{x}_i^j$ and the multiplication macro. Let
$Q(x_1, x_2, \ldots, x_k)$ be a monomial of degree $d$. We show how to
proceed by induction on $d$. If $d = 0$, then this is trivial. For
larger degrees, we evaluate $Q$ without its coefficient $c$, and then
apply macro ``\multip{\vr{s}}{c}{\vr{t}}''. If $d = 1$, then we can
simply transfer the appropriate counter $\vr{x}_i^j$. Otherwise, $Q$
is a product of two monomials $Q'$ and $Q''$ of smaller degrees. By
induction hypothesis, we can construct three copies of both monomials
$Q'$ and $Q''$. Then, using the multiplication macro, we obtain
$Q$. Having evaluated all monomials, we can transfer each of their
values to a common counter using the \textbf{transfer} and
\textbf{remove} macros depending on whether their sign is positive or
negative. Finally, we test if the resulting value equals zero, which
corresponds to having a solution to $P(x_1, x_2, \ldots, x_k) = 0$.

Before arguing correctness of the construction, let us see why the
whole program can be translated as a monogenic affine $\Z$-VASS $\V$,
\ie\ using only the identity matrix and one extra matrix
$\mat{A}$. First, note that all macros have internal counters
$\vr{x}$, $\vr{y}$ and $\vr{z}$. Every time we use a macro, we use
three fresh counters, increasing the dimension of $\V$. Second, note
that the only macro that requires a matrix different from the identity
is ``\square{\vr{s}}{\vr{t}}'', which doubles $\vr{y}$ during an
assignment to $\vr{z}$ at line~\ref{l:e}. We will construct the matrix
$\mat{A}$ as follows. Suppose we want to encode one of the squaring
macros. Matrix $\mat{A}$ will have the same updates for all counters
denoted as $\vr{z}$ in all macros, but the vector will use constants
according to this macro. That is, all coordinates for counters not
occurring in this macro will be $0$. In particular, counter $\vr{z}$
from this macro will be updated like in line~\ref{l:e},
\ie\ ``\add{\vr{z}}{2\vr{y} + 1}'', and all other counters
corresponding to some other $\vr{z}$ will be updated as
``\add{\vr{z}}{2\vr{y}}''.

\parag{Correctness.} We conclude by proving correctness of the
construction. If $P$ has a solution, then it is straightforward to
extract a run from $\V$: (a)~each $\vr{x}_i^j$ is initialized
according to the solution; and (b)~each loop of the program is
performed the exact number times so that each zero-test holds. It
remains to observe that after performing a zero-test on some counter,
$\V$ does not perform any operation on this counter or performs
``\add{\vr{z}}{2\vr{y}}''. But if the values of $\vr{y}$ and $\vr{z}$
are equal to zero, then $\vr{z}$ will remain equal to zero after such
an update.

Conversely, suppose there is a reachability witness (from $\vec{0}$ to
$\vec{0}$). We claim that the initialization of counters $\vr{x}_i^j$
provides a solution to $P$. To prove this, it suffices to show that
every zero-test was valid. This is clear for all counters except for
the $\vr{z}$ within the squaring macro. Indeed, all other counters
never change their values afterwards. However, counter $\vr{z}$ is
updated by ``\add{\vr{z}}{2\vr{y}}''. If $\vr{y}$ is non zero, then
this will be detected by the zero-test on~$\vr{y}$. Otherwise, the
update ``\add{\vr{z}}{2\vr{y}}'' never changes the value of $\vr{z}$
as required.

\section{Conclusion} \label{sec:conclusion}

We have shown that the reachability problem for afmp-$\Z$-VASS reduces
to the reachability problem for $\Z$-VASS, \ie\ every afmp-$\Z$-VASS
$\V$ can be simulated by a $\Z$-VASS of size polynomial in $|\V|$,
$|\monoid[\V]|$ and $\norm{\monoid[\V]}$. In particular, this allowed
us to establish that the reachability relation of any afmp-$\Z$-VASS
is semilinear.

For all nonnegative classes and consequently for all of the variants
we studied --- reset, permutation, transfer, copy and copyless
$\Z$-VASS --- $|\monoid[\V]|$ and $\norm{\monoid[\V]}$ are
of exponential size, thus yielding a PSPACE upper bound on their
reachability problems. Moreover, we have established PSPACE-hardness
for all of these specific classes, except for the reset case which is
NP-complete.

We do not know whether an exponential bound on $\norm{\monoid_\V}$
holds for any class of afmp-$\Z$-VASS over $\Z^{d \times d}$. We are
aware that an exponential upper bound holds when $\monoid_\V$ is
generated by a single matrix~\cite{IS16}; and when $\monoid_\V$ is a
group then we have an exponential bound but only on $|\monoid_\V|$
(see~\cite{KP02} for an exposition on the group case).

Finally, we have shown that there exists a (monogenic) class without
the finite-monoid property for which reachability is decidable. This
result was complemented by showing that reachability is undecidable in
general for monogenic classes.

\section*{Acknowledgments}

We are thankful to James Worrell for insightful discussions on
transfer VASS.

\bibliographystyle{alpha}
\bibliography{bibliography}

\newcommand{\etalchar}[1]{$^{#1}$}
\begin{thebibliography}{FFSPS11}

\bibitem[ACJT96]{ACJT96}
Parosh~Aziz Abdulla, Karlis Cerans, Bengt Jonsson, and Yih{-}Kuen Tsay.
\newblock General decidability theorems for infinite-state systems.
\newblock In {\em Proc.\ $11^\text{th}$ Annual {IEEE} Symposium on Logic in
  Computer Science ({LICS})}, pages 313--321, 1996.

\bibitem[AD16]{AD16}
Parosh~Aziz Abdulla and Giorgio Delzanno.
\newblock Parameterized verification.
\newblock {\em International Journal on Software Tools for Technology
  Transfer}, 18(5):469--473, 2016.

\bibitem[AFR14]{AFR14}
Rajeev Alur, Adam Freilich, and Mukund Raghothaman.
\newblock Regular combinators for string transformations.
\newblock In {\em Proc.\ Joint Meeting of the $23^\text{rd}$ {EACSL} Annual
  Conference on Computer Science Logic ({CSL}) and the $29^\text{th}$
  {ACM/IEEE} Symposium on Logic in Computer Science ({LICS})}, pages 9:1--9:10,
  2014.

\bibitem[AK76]{AK76}
Toshiro Araki and Tadao Kasami.
\newblock Some decision problems related to the reachability problem for
  {P}etri nets.
\newblock {\em Theoretical Computer Science}, 3(1):85--104, 1976.

\bibitem[ALW16]{ALW16}
Konstantinos Athanasiou, Peizun Liu, and Thomas Wahl.
\newblock Unbounded-thread program verification using thread-state equations.
\newblock In {\em Proc.\ $8^\text{th}$ International Joint Conference on
  Automated Reasoning ({IJCAR})}, pages 516--531, 2016.

\bibitem[AR13]{AR13}
Rajeev Alur and Mukund Raghothaman.
\newblock Decision problems for additive regular functions.
\newblock In {\em Proc.\ $40^\text{th}$ International Colloquium on Automata,
  Languages, and Programming ({ICALP})}, pages 37--48, 2013.

\bibitem[BFG{\etalchar{+}}15]{BFGHM15}
Michael Blondin, Alain Finkel, Stefan G{\"{o}}ller, Christoph Haase, and Pierre
  McKenzie.
\newblock Reachability in two-dimensional vector addition systems with states
  is {PSPACE}-complete.
\newblock In {\em Proc.\ $30^\text{th}$ Annual {ACM/IEEE} Symposium on Logic in
  Computer Science ({LICS})}, pages 32--43, 2015.

\bibitem[BH17]{BH17}
Michael Blondin and Christoph Haase.
\newblock Logics for continuous reachability in {P}etri nets and vector
  addition systems with states.
\newblock In {\em Proc.\ $32^\text{nd}$ Annual {ACM/IEEE} Symposium on Logic in
  Computer Science ({LICS})}, pages 1--12, 2017.

\bibitem[BHM18]{BlondinHM18}
Michael Blondin, Christoph Haase, and Filip Mazowiecki.
\newblock Affine extensions of integer vector addition systems with states.
\newblock In {\em Proc.\ $29^\text{th}$ International Conference on Concurrency
  Theory ({CONCUR})}, pages 14:1--14:17, 2018.

\bibitem[Boi98]{Boi98}
Bernard Boigelot.
\newblock {\em Symbolic Methods for Exploring Infinite State Spaces}.
\newblock PhD thesis, Universit\'{e} de Li\`{e}ge, Belgium, 1998.

\bibitem[Bon13]{Bon13}
R{\'e}mi Bonnet.
\newblock {\em Theory of Well-Structured Transition Systems and Extended
  Vector-Addition Systems}.
\newblock PhD thesis, {\'{E}}cole normale sup{\'{e}}rieure de Cachan, France,
  2013.

\bibitem[CFM12]{CadilhacFM12}
Micha{\"{e}}l Cadilhac, Alain Finkel, and Pierre McKenzie.
\newblock Bounded {P}arikh automata.
\newblock {\em International Journal of Foundations of Computer Science},
  23(8):1691--1710, 2012.

\bibitem[CFM13]{CadilhacFM13}
Micha{\"{e}}l Cadilhac, Alain Finkel, and Pierre McKenzie.
\newblock Unambiguous constrained automata.
\newblock {\em International Journal of Foundations of Computer Science},
  24(7):1099--1116, 2013.

\bibitem[CH16]{CH16}
Dmitry Chistikov and Christoph Haase.
\newblock The taming of the semi-linear set.
\newblock In {\em Proc.\ $43^\text{rd}$ International Colloquium on Automata,
  Languages, and Programming ({ICALP})}, pages 128:1--128:13, 2016.

\bibitem[CLL{\etalchar{+}}19]{CzerwinskiLLLM19}
Wojciech Czerwinski, Slawomir Lasota, Ranko Lazic, J{\'{e}}r{\^{o}}me Leroux,
  and Filip Mazowiecki.
\newblock The reachability problem for petri nets is not elementary.
\newblock In {\em Proc.\ $51^\text{st}$ Annual {ACM} {SIGACT} Symposium on
  Theory of Computing ({STOC})}, pages 24--33, 2019.

\bibitem[Del16]{Del16}
Giorgio Delzanno.
\newblock A unified view of parameterized verification of abstract models of
  broadcast communication.
\newblock {\em International Journal on Software Tools for Technology
  Transfer}, 18(5):475--493, 2016.

\bibitem[DFPS98]{DFS98}
Catherine Dufourd, Alain Finkel, and {Ph}ilippe Schnoebelen.
\newblock Reset nets between decidability and undecidability.
\newblock In {\em Proc.\ $25^\text{th}$ International Colloquium on Automata,
  Languages and Programming {(ICALP})}, pages 103--115, 1998.

\bibitem[ELM{\etalchar{+}}14]{ELMMN14}
Javier Esparza, Rusl{\'{a}}n Ledesma{-}Garza, Rupak Majumdar, Philipp~J. Meyer,
  and Filip Niksic.
\newblock An {SMT}-based approach to coverability analysis.
\newblock In {\em Proc.\ $26^\text{th}$ International Conference on Computer
  Aided Verification ({CAV})}, pages 603--619, 2014.

\bibitem[EN98]{EN98}
E.~Allen Emerson and Kedar~S. Namjoshi.
\newblock On model checking for non-deterministic infinite-state systems.
\newblock In {\em Proc.\ $13^\text{th}$ Annual {ACM/IEEE} Symposium on Logic in
  Computer Science ({LICS})}, pages 70--80, 1998.

\bibitem[Esp98]{Esparza98}
Javier Esparza.
\newblock Decidability and complexity of {P}etri net problems --- an
  introduction.
\newblock In {\em Lectures on {P}etri Nets {I}}, pages 374--428, 1998.

\bibitem[FFSPS11]{FFSS11}
Diego Figueira, Santiago Figueira, Sylvain Schmitz, and {Ph}ilippe Schnoebelen.
\newblock {A}ckermannian and primitive-recursive bounds with {D}ickson's lemma.
\newblock In {\em Proc.\ $26^\text{th}$ Annual {IEEE} Symposium on Logic in
  Computer Science ({LICS})}, pages 269--278, 2011.

\bibitem[FGH13]{FinkelGH13}
Alain Finkel, Stefan G{\"{o}}ller, and Christoph Haase.
\newblock Reachability in register machines with polynomial updates.
\newblock In {\em Proc.\ $38^\text{th}$ International Symposium on Mathematical
  Foundations of Computer Science ({MFCS})}, pages 409--420, 2013.

\bibitem[FL02]{FL02}
Alain Finkel and J{\'{e}}r{\^{o}}me Leroux.
\newblock How to compose {P}resburger-accelerations: Applications to broadcast
  protocols.
\newblock In {\em Proc.\ $22^\text{nd}$ Conference on Foundations of Software
  Technology and Theoretical Computer Science ({FSTTCS})}, pages 145--156,
  2002.

\bibitem[HH14]{HH14}
Christoph Haase and Simon Halfon.
\newblock Integer vector addition systems with states.
\newblock In {\em Proc.\ $8^\text{th}$ International Workshop on Reachability
  Problems ({RP})}, pages 112--124, 2014.

\bibitem[HP79]{HP79}
John~E. Hopcroft and Jean{-}Jacques Pansiot.
\newblock On the reachability problem for 5-dimensional vector addition
  systems.
\newblock {\em Theoretical Computer Science}, 8:135--159, 1979.

\bibitem[HU79]{HU79}
John~E. Hopcroft and Jeffrey~D. Ullman.
\newblock {\em Introduction to Automata Theory, Languages and Computation}.
\newblock Addison-Wesley, 1979.

\bibitem[IS16]{IS16}
Radu Iosif and Arnaud Sangnier.
\newblock How hard is it to verify flat affine counter systems with the finite
  monoid property?
\newblock In {\em Proc.\ $14^\text{th}$ International Symposium on Automated
  Technology for Verification and Analysis ({ATVA})}, pages 89--105, 2016.

\bibitem[KKW14]{KKW14}
Alexander Kaiser, Daniel Kroening, and Thomas Wahl.
\newblock A widening approach to multithreaded program verification.
\newblock {\em {ACM} Transactions on Programming Languages and Systems},
  36(4):14:1--14:29, 2014.

\bibitem[KM69]{KM69}
Richard~M. Karp and Raymond~E. Miller.
\newblock Parallel program schemata.
\newblock {\em Journal of Computer and System Sciences}, 3(2):147--195, 1969.

\bibitem[Kos82]{Kos82}
S.~Rao Kosaraju.
\newblock Decidability of reachability in vector addition systems (preliminary
  version).
\newblock In {\em Proc.\ $14^\text{th}$ Annual {ACM} Symposium on Theory of
  Computing ({STOC})}, pages 267--281, 1982.

\bibitem[KP02]{KP02}
James Kuzmanovich and Andrey Pavlichenkov.
\newblock Finite groups of matrices whose entries are integers.
\newblock {\em The American Mathematical Monthly}, 109(2):173--186, 2002.

\bibitem[Ler12]{Ler12}
J{\'{e}}r{\^{o}}me Leroux.
\newblock Vector addition systems reachability problem (a simpler solution).
\newblock In {\em The Alan Turing Centenary Conference}, pages 214--228, 2012.

\bibitem[Lip76]{Lip76}
Richard~J. Lipton.
\newblock The reachability problem requires exponential space.
\newblock Technical Report~63, Department of Computer Science, Yale University,
  1976.

\bibitem[May84]{May84}
Ernst~W. Mayr.
\newblock An algorithm for the general {P}etri net reachability problem.
\newblock {\em {SIAM} Journal on Computing}, 13(3):441--460, 1984.

\bibitem[Min67]{M67}
Marvin~Lee Minsky.
\newblock {\em Computation: Finite and Infinite Machines}.
\newblock Prentice-Hall, 1967.

\bibitem[MS77]{MS77}
Arnaldo Mandel and Imre Simon.
\newblock On finite semigroups of matrices.
\newblock {\em Theoretical Computer Science}, 5(2):101--111, 1977.

\bibitem[Pre29]{Pr29}
Mojżesz Presburger.
\newblock Über die {V}ollständigkeit eines gewissen {S}ystems der
  {A}rithmetik ganzer {Z}ahlen, in welchem die {A}ddition als einzige
  {O}peration hervortritt.
\newblock {\em Comptes Rendus du {$\text{I}^\text{er}$} Congrès des
  mathématiciens des pays slaves}, pages 192--201, 1929.

\bibitem[Rac78]{Rac78}
Charles Rackoff.
\newblock The covering and boundedness problems for vector addition systems.
\newblock {\em Theoretical Computer Science}, 6:223--231, 1978.

\bibitem[Rei08]{Rei08}
Klaus Reinhardt.
\newblock Reachability in {P}etri nets with inhibitor arcs.
\newblock {\em Electronic Notes in Theoretical Computer Science}, 223:239--264,
  2008.

\bibitem[Rei15]{Rei15}
Julien Reichert.
\newblock {\em Reachability games with counters: decidability and algorithms}.
\newblock PhD thesis, {\'{E}}cole normale sup{\'{e}}rieure de Cachan, France,
  2015.

\bibitem[Sch10]{Sch10}
Philippe Schnoebelen.
\newblock Revisiting {A}ckermann-hardness for lossy counter machines and reset
  {P}etri nets.
\newblock In {\em Proc.\ $35^\text{th}$ International Symposium Mathematical
  Foundations of Computer Science ({MFCS})}, pages 616--628, 2010.

\bibitem[Web87]{We87}
Andreas Weber.
\newblock {\em \"{U}ber die Mehrdeutigkeit und Wertigkeit von endlichen
  Automaten und Transducern}.
\newblock PhD thesis, Goethe-Universität Frankfurt am Main, 1987.

\bibitem[WS91]{WS91}
Andreas Weber and Helmut Seidl.
\newblock On finitely generated monoids of matrices with entries in
  $\mathbb{N}$.
\newblock {\em Informatique Théorique et Applications}, 25:19--38, 1991.

\end{thebibliography}

\end{document}